\documentclass[a4paper,12pt,notitlepage]{article}
\pdfoutput=1 
\usepackage[utf8]{inputenc}
\usepackage{exscale}
\usepackage{amsmath}
\usepackage{amsthm}
\usepackage{amssymb}
\usepackage{graphicx}
\usepackage{xcolor}
\usepackage{enumerate}
\usepackage{xspace}
\usepackage{url}
\usepackage{color}
\usepackage{doi} 
\usepackage{fancyhdr}

\usepackage[round]{natbib}
\newcommand{\mybibstyle}{\bibliographystyle{abbrvnat}}
\usepackage{hyperref}
\hypersetup{  colorlinks=true,%
              citecolor=black,%
              filecolor=black,%
              linkcolor=black,%
              urlcolor=blue,%
}
\theoremstyle{plain}
\newtheorem{theorem}{Theorem}[section]
\newtheorem*{theorem*}{Theorem}
\newtheorem{lemma}[theorem]{Lemma}
\newtheorem*{lemma*}{Lemma}
\newtheorem{proposition}[theorem]{Proposition}
\newtheorem*{proposition*}{Proposition}
\newtheorem{corollary}[theorem]{Corollary}
\newtheorem*{corollary*}{Corollary}

\newtheorem*{condition*}{Condition}

\newtheorem*{definition*}{Definition}

\theoremstyle{remark}
\newtheorem{remark}[theorem]{Remark}
\newtheorem*{remark*}{Remark}

\newtheorem*{example*}{Example}

\newcommand{\mybeginproofof}[1]{\begin{proof}[Proof of {#1}]}
\newcommand{\mybeginproof}{\begin{proof}}
\newcommand{\myendproof}{\end{proof}}

\newlength{\mygraphicwidth}
\newlength{\mysubgraphicwidth}
\newcommand{\notiz}[1]{\relax}
\DeclareMathOperator{\esssup}{ess\,sup}
\DeclareMathOperator{\cov}{cov}
\newcommand{\abs}[1]{| #1 |}
\newcommand{\absfl}[1]{\left| #1 \right|}
\newcommand{\Acal}{\mathcal{A}}
\newcommand{\anbr}[1]{\langle #1 \rangle}
\newcommand{\anbrup}[1]{\langle #1 \rangle_{\mathrm{up}}}
\newcommand{\anbrfl}[1]{\left\langle #1 \right\rangle}
\newcommand{\asconv}{\stackrel{\mathrm{a.s.}}{\to}}
\newcommand{\At}{A_t}
\newcommand{\bnd}{\partial} 

\newcommand{\Bt}{B_t}
\newcommand{\ceilbr}[1]{\lceil #1 \rceil}

\newcommand{\Ccal}{\mathcal{C}}
\newcommand{\Cdn}{C_{d,n}}
\newcommand{\Cin}{C_{i,n}}
\newcommand{\Cn}{C_n}
\newcommand{\Cnast}{C_n^\ast}

\newcommand{\Conen}{C_{1,n}}
\newcommand{\cubr}[1]{\{{#1}\}}
\newcommand{\cubrfl}[1]{\left\{{#1}\right\}}
\newcommand{\Dcal}{\mathcal{D}}

\newcommand{\Dcaldeltazer}{\Dcal_{\delta_0}}
\newcommand{\dm}{\mathrm{d}}
\newcommand{\Dminus}{D_{-}}

\newcommand{\Dplus}{D_{+}}

\newcommand{\enquote}[1]{\lq\lq{}{#1}\rq\rq{}}
\newcommand{\eps}{\varepsilon}
\newcommand{\ExpDistr}{\mathrm{Exp}}
\newcommand{\Fi}{F_i}
\newcommand{\Fdn}{F_{d,n}}
\newcommand{\Fin}{F_{i,n}}

\newcommand{\Fnast}{F_n^\ast}

\newcommand{\Fonen}{F_{1,n}}
\newcommand{\Ftilde}{\widetilde{F}}
\newcommand{\Ftildein}{\Ftilde_{i,n}}

\newcommand{\Gbb}{\mathbb{G}}

\newcommand{\Gnast}{G_n^\ast}

\newcommand{\GPsinast}{G_{\Psi,n}^\ast}
\newcommand{\Gtilde}{\widetilde{G}}
\newcommand{\Gtilden}{\Gtilde_n}
\newcommand{\ginv}{^{\leftarrow}}
\newcommand{\half}{\frac{1}{2}}
\newcommand{\Hcal}{\mathcal{H}}
\newcommand{\Hcalbar}{\overline{\Hcal}}
\newcommand{\HcalT}{\mathcal{H}_T}
\newcommand{\Hcald}{\Hcal_d}
\newcommand{\Hcaldbar}{\overline{\Hcald}}

\newcommand{\Hcaltwobar}{\overline{\Hcal_2}}

\newcommand{\Iast}{I_\ast}

\newcommand{\id}{\mathrm{id}}

\newcommand{\Iminus}{I_{-}}

\newcommand{\inv}{^{-1}}
\newcommand{\Iplus}{I_{+}}
\newcommand{\linf}{l^{\infty}}

\newcommand{\mylefteqn}{\hspace{2em}&\hspace{-2em}}
\newcommand{\N}{\Nbb}
\newcommand{\Nbb}{\mathbb{N}}
\newcommand{\Nbr}{N_{[\,]}}
\newcommand{\Ncal}{\mathcal{N}}
\newcommand{\norm}[1]{\Vert #1 \Vert}
\newcommand{\normfl}[1]{\left\Vert #1 \right\Vert}
\newcommand{\oneby}[1]{\frac{1}{#1}}
\newcommand{\onebyn}{\oneby{n}}
\newcommand{\Pbb}{\mathbb{P}}
\newcommand{\pospart}{_{+}}
\newcommand{\powminustheta}{^{-\theta}}
\newcommand{\Prob}{\mathrm{P}}
\newcommand{\ProbC}{\Prob_{C}}
\newcommand{\ProbCn}{\Prob_{\Cn}}
\newcommand{\ProbCnast}{\Prob_{\Cnast}}

\newcommand{\ProbCrho}{\Prob_{C_\rho}}
\newcommand{\ProbCtheta}{\Prob_{C_\theta}}
\newcommand{\ProbFnast}{\Prob_{\Fnast}}
\newcommand{\R}{\Rbb}
\newcommand{\Rd}{\R^d}
\newcommand{\Rbb}{\mathbb{R}}
\newcommand{\Rcal}{\mathcal{R}}
\newcommand{\robr}[1]{(#1)}
\newcommand{\robrfl}[1]{\left( #1 \right)}
\newcommand{\sampj}{^{(j)}}
\newcommand{\sampk}{^{(k)}}
\newcommand{\sampn}{^{(n)}}
\newcommand{\sampone}{^{(1)}}
\newcommand{\secondrevision}{\relax}
\newcommand{\sqbr}[1]{[#1]}
\newcommand{\sqbrfl}[1]{\left[#1\right]}
\newcommand{\sqrtn}{\sqrt{n}}
\newcommand{\sumioned}{\sum_{i=1}^d}
\newcommand{\sumjonen}{\sum_{j=1}^n}
\newcommand{\sumkonen}{\sum_{k=1}^n}
\newcommand{\symdiff}{\bigtriangleup}
\newcommand{\Tcal}{\mathcal{T}}

\newcommand{\Tcald}{\Tcal_d}
\newcommand{\Tcaldbar}{\overline{\Tcal_d}}

\newcommand{\Tcaltwobar}{\overline{\Tcal_2}}
\newcommand{\Tn}{T_n}
\newcommand{\toinf}{\to\infty}
\newcommand{\tr}{^{\top}}
\newcommand{\Ttilde}{\widetilde{T}}
\newcommand{\Ttilden}{\Ttilde_n}
\newcommand{\uast}{u^\ast}
\newcommand{\ubar}{\overline{u}}
\newcommand{\Udelta}{U_\delta}
\newcommand{\Udeltaone}{U_{\delta_1}}
\newcommand{\uintd}{[0,1]^d}
\newcommand{\uintepsd}{[\eps, 1-\eps]^d}
\newcommand{\ulbar}{\underline{u}}
\newcommand{\upd}{^{(d)}}
\newcommand{\upi}{^{(i)}}
\newcommand{\upiminusone}{^{(i-1)}}
\newcommand{\upj}{^{(j)}}
\newcommand{\upone}{^{(1)}}
\newcommand{\uptwo}{^{(2)}}
\newcommand{\upzer}{^{(0)}}

\newcommand{\vbarast}{{\overline{v}}^\ast}
\newcommand{\vlbarast}{{\underline{v}}^\ast}
\newcommand{\VC}{Vapnik--\u{C}ervonenkis\xspace}
\newcommand{\Vdeltat}{V_{\delta,t}}
\newcommand{\Wdeltat}{W_{\delta,t}}
\newcommand{\weakconv}{\stackrel{\mathrm{w}}{\to}}
\newcommand{\Xtilde}{\widetilde{X}}
\newcommand{\Ytilde}{\widetilde{Y}}
\newcommand{\Ytilden}{\Ytilde_n}
\begin{document}
\date{August 10, 2015}
\title
{Risk aggregation with empirical margins: Latin hypercubes, empirical copulas, and convergence of sum distributions}

\author{
Georg Mainik 
\footnote{RiskLab, Department of Mathematics, ETH Zurich; \  \href{http://www.georgmainik.com}{www.georgmainik.com}} 
}
\maketitle
\vspace{-0.5em}
\begin{abstract}
%
This paper studies convergence properties of multivariate distributions 
constructed by endowing empirical margins with a copula. 
This setting includes 
Latin Hypercube Sampling with dependence, also 
known as the Iman--Conover method.  
The primary question addressed here is the convergence of the component sum, 
which is relevant to risk aggregation in insurance and finance. 
This paper shows that a CLT for the aggregated risk 
distribution is not available, so that the underlying mathematical problem 
goes beyond classic functional CLTs for empirical copulas.  
%
This issue is relevant to Monte-Carlo based risk aggregation in 
all multivariate models generated by plugging empirical margins into 
a copula. 
Instead of a functional CLT,
this paper establishes strong uniform consistency of 
the estimated sum distribution function and provides a sufficient criterion 
for the convergence rate $O(n^{-1/2})$ in probability.  
These convergence results hold for all copulas with bounded densities. 
Examples with unbounded densities 
include bivariate Clayton and Gauss copulas. 
The convergence results are not specific to the component 
sum and hold also for any other componentwise non-decreasing aggregation 
function.
%
%
On the other hand, convergence of estimates for the joint distribution 
is much easier to prove, including CLTs. 
Beyond Iman--Conover estimates, the results of this paper 
apply to multivariate distributions obtained by 
plugging empirical margins into 
an exact copula or by plugging exact margins into an empirical copula. 
%
%

\\[0.5em]
\end{abstract}
\thispagestyle{fancy}

\section{Introduction}\label{sec:1}
In various real-world 
applications, multivariate stochastic models are 
constructed upon empirical marginal data and an assumption on the dependence 
structure between the margins. This dependence assumption is often formulated 
in terms of copulas. 
The major reason for this set-up is the lack of 
multivariate data sets, as it is often the case in insurance and finance.
This approach may appear artificial from the statistical point of view, 
but it arises naturally in the context of stress testing. 
In addition to finance and insurance, relevant application areas 
include engineering and environmental studies. 
Sometimes the marginal data is not even based on observations, 
but is generated by a univariate model that is considered reliable. 
Many of these models are so complex that the 
resulting distributions cannot be expressed analytically. 
In such cases exact marginal distributions are replaced by 
empirical distributions of simulated univariate samples. 
These empirical margins are endowed with some dependence structure  
to obtain a multivariate distribution. The computation of aggregated 
risk or other characteristics of this multivariate model is 
typically based on Monte-Carlo techniques. 
\par
\subsection*{Iman--Conover: dependence \enquote{injection} by sample reordering}
Related methods include generation of synthetic multivariate samples 
from univariate data sets. 
Whilst the margins of such a synthetic sample accord with the univariate data, 
its dependence structure is modified to fit the application's needs. 
The most basic example is the classic 
Latin Hypercube Sampling method, which mimics independent margins. 
It is a popular tool for removing spurious correlations from multivariate 
data sets. 
This method is also applied to 
variance reduction in the simulation of independent random variables  
\citep[cf.][]{McKay/Beckman/Conover:1979, Stein:1987, Owen:1992, Iman:2008}. 
Similar applications to dependent random variables include 
variance reduction in Monte-Carlo methods \citep{Packham/Schmidt:2010} 
and in copula estimation \citep{Genest/Segers:2010}. 
\par 
An extension of Latin Hypercube Sampling that brings dependence 
into the samples was proposed by \citet{Iman/Conover:1982}. 
The original description of the Iman--Conover method uses random reordering of 
marginal samples, and the intention there was to control the 
rank correlations in the synthetic multivariate sample. 
The reordering is performed according to the vectors of marginal ranks 
in an i.i.d.\ sample of some multivariate distribution, say, $H$, 
with continuous margins.
Thus rank correlations of $H$ are \enquote{injected} into the synthetic sample. 
This procedure is equivalent 
to plugging 
empirical margins
(obtained from asynchronous observations)  
into the rank based empirical 
copula of a sample of $H$ \citep{Arbenz/Hummel/Mainik:2012}. 
Moreover, it turned out that the Iman--Conover method allows to introduce  
not only the rank correlations of $H$ into the synthetic samples, 
but the entire copula of $H$ 
\citep[cf.][]{Arbenz/Hummel/Mainik:2012,Mildenhall:2005}. 
%
In somewhat weaker sense, 
these results are related to the approximation of stochastic dependence by 
deterministic functions and to the pioneering result 
by~\citet{Kimeldorf/Sampson:1978}. Further developments in that area 
include measure preserving transformations \citep{Vitale:1990} and 
shuffles of $\min$ \citep{Durante/Sarkoci/Sempi:2009}. 
In statistical optimization, reordering techniques were also used by
\cite{Rueschendorf:1983}. 
A very recent, related application in 
quantitative risk management is a rearrangement algorithm 
that computes worst-case bounds for 
the aggregated loss quantiles in a portfolio with given marginal 
distributions \citep[cf.][and references therein]{Embrechts/Puccetti/Rueschendorf:2013}.  
\par
Using explicit reorderings of univariate marginal samples, the 
Iman--Conover method has a unique algorithmic tractability.
It is implemented in various software packages, and it serves as 
a standard tool in dependence modelling and uncertainty analysis. 
The reordering algorithm allows even to construct synthetic samples 
with hierarchical dependence structures 
that meet the needs of 
risk aggregation in insurance and reinsurance companies
\citep{Arbenz/Hummel/Mainik:2012}. 
The distribution of the aggregated risk is estimated by the empirical 
distribution of the component sums 
$\Xtilde_1\sampk+\ldots+\Xtilde_d\sampk$ 
of the synthetic samples $\Xtilde\sampk=\robr{\Xtilde_1\sampk,\ldots,\Xtilde_d\sampk}$ for $k=1,\ldots,n$. 
This Monte-Carlo approach has computational advantages. 
The resulting convergence rate of $n^{-1/2}$ 
(or even faster with Quasi-Monte-Carlo using special sequences) 
allows to outperform explicit calculation of sum distributions 
already for moderate dimensions $d\ge4$
\citep[cf.][]{Arbenz/Embrechts/Puccetti:2011}. 
\par
\subsection*{Challenge and contribution: convergence proofs}
Despite its popularity,  
some applications of the Iman--Conover method have been  
justified by simulations rather than by mathematical proofs.
The original publication  \citep{Iman/Conover:1982} 
derives its conclusions from 
promising simulation results for the distribution of 
the following function of a $4$-dimensional random vector: 
$f(X_1,\ldots,X_4) = X_1 + X_2(X_3 - \log\abs{X_1}) + \exp(X_4/4)$.  
Yet a rigorous proof is still missing. 
The present paper provides a convergence proof for Iman--Conover estimates 
of the component sum distribution. 
It also includes a proof sketch for the much simpler case of 
the estimated joint distribution. 
Both problems have been open until now. 
%
\par
The solutions given in this paper 
are derived from 
the empirical process 
theory as presented in~\citet{van_der_Vaart/Wellner:1996}. 
Under appropriate regularity assumptions, 
Iman--Conover estimates of the sum distribution are 
strongly uniformly consistent with  convergence rate 
$O_\Prob(n^{-1/2})$ (see Theorems~\ref{thm:1} and \ref{thm:2}).
The convergence of Iman--Conover estimates for the joint distribution 
is discussed in cf.\ Remark~\ref{rem:9}. 
All these findings are not specific to the component sum and extend immediately 
to all componentwise non-decreasing functions 
(see Corollary~\ref{cor:3}). 
Moreover, Theorems~\ref{thm:1} and \ref{thm:2} 
also cover the convergence of aggregated risk 
distributions obtained by Monte-Carlo sampling of a multivariate model 
constructed by plugging empirical margins into a copula 
(see Remark~\ref{rem:7}). 
In fact, both sampling methods (reordering by Iman--Conover and classic 
top-down sampling with empirical margins instead of the exact ones) 
lead to the same mathematical problem. 
This is discussed in Remark~\ref{rem:2}(\ref{item:rem.2.c}).  
\par
The regularity assumptions used here to establish the $O_\Prob(n^{-1/2})$ 
convergence rate for Iman--Conover estimates of sum distributions 
are satisfied for all copulas with bounded 
densities. This case includes the independence copula in arbitrary dimension 
$d\ge2$. 
The assumptions are also satisfied for all bivariate Clayton copulas 
and for bivariate Gauss copulas with correlation parameter $\rho\ge0$. 
The convergence rate for $\rho<0$ is, if at all, only slightly weaker. 
The best bound that is currently available for $\rho<0$  
is $O_\Prob(n^{-1/2}\sqrt{\log n})$. 
\par
The regularity assumptions for the marginal distributions 
involved in the Iman--Conover method are absolutely natural, and they are 
always satisfied by empirical distribution of i.i.d.\ samples: 
Strong uniform consistency of Iman--Conover estimates 
needs strong uniform consistency of 
consistency of empirical margins, 
whereas the uniform $O_\Prob(n^{-1/2})$ convergence rate of 
Iman--Conover requires the same uniform convergence rate of $O_\Prob(n^{-1/2})$ 
in the margins. 
\par
\subsection*{Why a precise CLT remains elusive}
The convergence results obtained here are related to standard convergence 
results for empirical copulas 
\citep[cf.][]{Rueschendorf:1976,Deheuvels:1979,Fermanian/Radulovic/Wegkamp:2004,Segers:2012}. 
However, the mathematical problem for the sum distribution goes beyond the 
standard setting, where empirical measures are evaluated on rectangular sets. 
In the case of sum distributions, the usage of empirical margins in the 
construction of the multivariate model 
significantly extends the class of sets 
on which the empirical process of the copula sample should converge. 
%
As shown in Section~\ref{sec:2a}, the 
canonical way to prove asymptotic normality for the 
Iman--Conover estimator of the sum distribution 
would need a uniform CLT for the copula sample on the 
collection of so-called lower layers in $\uintd$. 
However, this class is too complex for a uniform CLT 
\citep[cf.][Theorems 8.3.2, 12.4.1, and 12.4.2]{Dudley:1999}. 
For this reason the proofs of consistency and convergence rate 
presented here sacrifice the precise asymptotic variance and use approximations
that allow to simplify the problem. 
%
This technical difficulty is not specific to the Iman--Conover method.  
It also arises in any other application where multivariate 
samples are generated from a simulated copula sample and empirical marginal 
distributions. 
As mentioned above, this approach is very popular in practice, 
especially for computational reasons. 
Similar problems also arise in applications that combine exact marginal distributions 
with empirical copulas. 
\par
\subsection*{Structure of the paper}
The paper is organized as follows. Section~\ref{sec:2} introduces the  
reordering method and highlights the relations between sample reordering and 
empirical copulas. 
The complexity issues are discussed in Section~\ref{sec:2a}. 
The convergence results are established in 
Section~\ref{sec:3}. The underlying regularity assumptions 
are discussed in Section~\ref{sec:4}, including examples of copula 
families that satisfy them. 
Conclusions are stated in Section~\ref{sec:5}.
%
%
\section{Empirical copulas and sample reordering}\label{sec:2}
\par
Let $X=(X_1,\ldots,X_d)$ be a random vector in $\Rd$ with 
joint distribution function $F$, marginal distribution functions 
$F_1,\ldots,F_d$, and copula $C$. That is, 
\begin{align}\label{eq:052}
F(x) = C \robr{F_1(x_1),\ldots,F_d(x_d)},
\end{align}
where $C$ is a probability distribution function on $\uintd$ with uniform margins. 
By Sklar's Theorem, 
any multivariate distribution function 
$F$ admits this representation. 
\par
We assume throughout the following that $\Fi$ are unknown 
and that we have some uniform approximations 
$\Fin$, $i=1,\ldots,d$.
The true margins $\Fi$ need not be continuous. 
In this case the representation~\eqref{eq:052} is not unique, 
but it is not an issue in our application, which is rather computational 
than statistical. 
As sketched in the Introduction, we consider the case where only univariate, 
asynchronous observations of the components $X_i$ are available, 
and the copula $C$ is set by expert judgement to compute the 
resulting distribution of the component sum. 
In practice, the choice of the copula $C$ aims at dependence 
characteristics that are known or assumed for the random vector $X$. 
This choice also depends on the ability to sample $C$ or any other 
multivariate distribution with continuous margins and copula $C$. 
\par
To keep the presentation simple, we assume that $\Fin$ 
are empirical distribution functions of some 
univariate samples $X_i\sampone,\ldots,X_i\sampn$ for $i=1,\ldots,d$:
\begin{align} \label{eq:014}
\Fin(t):=\onebyn\sumkonen 1\cubrfl{X_i\sampk \le t} 
,\quad t\in\R.
\end{align}
These samples need not be i.i.d. We will only assume that 
$\norm{\Fin - \Fi}_\infty \to 0$, either $\Prob$-a.s.\ or in probability. 
Extensions to the general case will be given later on. 
\par
Let $X_i^{(1:n)} \le \ldots \le X_i^{(n:n)}$ 
denote the order statistics of the $i$-th component $X_i$ 
for $i=1,\ldots,d$,  
and   
let $\Prob_H$ denote the probability measure 
with distribution function $H$. 
The Iman--Conover method approximates $\Prob_F$ 
by the empirical measure of the following 
synthetic multivariate sample:  
\begin{align} \label{eq:047}
\Xtilde\sampj:=
\robrfl{X_1^{\robrfl{R_1\sampj:n}},\ldots,X_d^{\robrfl{R_d\sampj:n}}}
, \quad j=1,\ldots,n,
\end{align}
where $R_i\sampone,\ldots,R_i\sampn$ for $i=1,\ldots,d$ are the marginal ranks of a simulated i.i.d.\ sample $U\sampone,\ldots,U\sampn \sim C$: 
\[
R_i\sampj = \sumkonen 1\cubrfl{U_i\sampj \ge U_i\sampk}
,\quad i=1,\ldots,d,\ j=1,\ldots,n. 
\]
It is easy to verify \citep[cf.][Theorem 3.2]{Arbenz/Hummel/Mainik:2012} that 
the empirical distribution function of the synthetic sample~\eqref{eq:047}  
is equal to  
\begin{align} \label{eq:008}
\Fnast(x) := \Cnast\robr{\Fonen(x_1),\ldots,\Fdn(x_d)}, 
\end{align}
where $\Cnast$ is the rank based \emph{empirical copula} of  
$U\sampone,\ldots,U\sampn$: 
\begin{align}\label{eq:063}
\Cnast(u) := \onebyn \sumkonen 
1\cubrfl{\onebyn R_1\sampk \le u_1,\ldots, \onebyn R_d\sampk \le u_d}.
\end{align}
This links the convergence of the Iman--Conover method 
to the convergence of $\Prob_{\Fnast}$ to $\Prob_F$, 
and hence to the convergence of $\ProbCnast$ to $\ProbC$. 
\par
{\secondrevision%
\begin{remark}
\begin{enumerate}[(a)]
\item
The central application of the Iman--Conover method discussed in the present 
paper is the computation of the aggregated risk distribution. 
The most common risk aggregation function is the sum. 
In this case one must compute
the probability distribution of the random variable 
$\sumioned X_i$ for $(X_1,\ldots,X_d)\sim F$ with $F$ defined 
in~\eqref{eq:052}. 
Iman--Conover involves two approximations: replacing 
the unknown margins $F_i$ by their empirical versions $\Fin$, 
and replacing the known (or treated as known) copula $C$ by its 
empirical version $\Cnast$. 
\item
Using $\Cnast$ may appear unnecessary because 
one can also compute the sum distribution for a random vector 
with margins $\Fin$ and exact copula $C$. However, computation of sum 
distributions from margins and copulas is quite difficult in practice.  
It involves numeric integration on non-rectangular sets, which cannot be 
reduced to taking the value of 
$C\robrfl{\Fonen(x_1),\ldots,\Fdn(x_n)}$ for a few points $x=(x_1,\ldots,x_d)$.
Implementations of this kind are exposed to the curse of 
dimensions. Monte-Carlo methods, which Iman--Conover belongs to, 
have the convergence rate of $1/\sqrtn$, and Quasi-Monte-Carlo 
methods using special sequences may even allow to achieve the rate $1/n$.  
According to 
\cite{Arbenz/Embrechts/Puccetti:2011}, 
explicit computation of sum distributions is outperformed by 
Monte-Carlo already for $d=4$. 
\item
Another motivation of the Iman--Conover method is its flexibility and  
algorithmic tractability. 
It only includes reordering of samples and works in the same way for any 
dimension. 
Moreover, sample reordering is compatible with 
hierarchical dependence structures that can be described as trees with 
univariate distributions in leaves and copulas in branching nodes
\citep[cf.][]{Arbenz/Hummel/Mainik:2012}.
In each branching node, the marginal distributions are aggregated according 
to the node's copula and the resulting aggregated (typically, sum) distribution 
is propagated to the next aggregation level.
As shown in~\cite{Arbenz/Hummel/Mainik:2012}, sample reordering can be 
implemented for a whole tree. 
The setting with one copula and margins discussed in the present paper is 
the basic element of such aggregation trees. 
The results presented here allow to prove the convergence of the 
aggregated (say, sum) distribution in every tree node, including 
the total sum.  
\end{enumerate}
\end{remark}
\par
Now let us return to the technical details of the Iman--Conover estimator 
for the aggregated sum distribution.
}%
As the random variables $U_i\sampk$ are continuously distributed, 
they have no ties $\Prob$-a.s.
Thus $\ProbCnast$ consists 
$\Prob$-a.s.\ of $n$ atoms of size $1/n$.
Moreover, these $n$ atoms 
build a Latin hypercube 
on the $d$-variate grid $\cubr{\frac{1}{n},\frac{2}{n},\ldots,1}^d$, 
i.e., each section $\cubr{x\in\cubr{\frac{1}{n},\frac{2}{n},\ldots,1}^d: x_i=\frac{j}{n}}$ for 
$i=1,\ldots,d$ and $j=1,\ldots,n$  contains precisely one atom. 
Therefore the Iman--Conover method is also called Latin Hypercube Sampling with dependence.  
\par
For $X\sim F$, let $G$ denote the distribution function of the component sum: 
$G(t) := \Prob\robr{X_1+\ldots+X_d \le t}$, $t\in \R$. 
The relation between $G$ and $\ProbC$ can be expressed as follows.
\begin{lemma} \label{lem:5}
\begin{align} \label{eq:062}
\forall t\in\R \quad G(t) = \ProbC(\anbr{T(\At)}) 
\end{align}
where $\At:=\cubr{x\in\Rd: \sumioned x_i \le t}$, $T(x):=\robr{F_1(x_1),\ldots,F_d(x_d)}$, and 
\[
\anbr{B}:=\cup_{u\in B}[0,u]
\quad
\text{for } B\subset\uintd.
\]
\end{lemma}
The notation $[0,u]$ refers to the closed $d$-dimensional interval between $0$ and $u$: $[0,u]:=[0,u_1]\times\cdots\times[0,u_d]$.
\begin{proof}
Let $U\sim C$ and denote 
\[
T\ginv(u):=\robr{F_1\ginv(u_1),\ldots,F_d\ginv(u_d)}
\] 
for $u\in\uintd$, where $F_i\ginv(y):=\inf\cubr{t\in\R: \Fi(t)\ge y}$ is the quantile function of $F_i$. 
%
It is well known that $T\ginv(U)\sim F$. 
Hence
\begin{align*} 
G(t) = \Prob_F(\At) = \Prob\robrfl{T\ginv(U)\in\At} = 
\ProbC\robrfl{\cubrfl{u\in\uintd : T\ginv(u) \in \At}},
\end{align*}
and it suffices to show that $v\in\anbr{T(\At)}$ is equivalent to $T\ginv(v)\in\At$. 
\par
If $T\ginv(v)\in\At$, 
then $T\circ T\ginv(v)\in\anbr{T(\At)}$. 
Due to $F_i\circ F_i\ginv(v_i) \ge v_i$ for all $i$ 
this implies that $v\in\sqbr{0,T\circ T\ginv(v)} \subset \anbr{T(\At)}$. 
\par
If $v\in\anbr{T(\At)}$, then $v\le T(x)$ (componentwise) for some $x\in\At$. 
%
Since $F_i\ginv\circ F(x_i) \le x_i$ for all $i$, this yields  
$T\ginv(v) \le x$. As the function $x\mapsto\sumioned x_i$ is 
componentwise non-decreasing, we obtain that $T\ginv(v)\in\At$. 
\end{proof}
\begin{remark} \label{rem:1}
\begin{enumerate}[(a)]
\item 
The measurability of $\anbr{T(\At)}$ follows from the equivalence of $v\in\anbr{T(\At)}$ and $T\ginv(v)\in\At$. 
\item 
The purpose of the operator $\anbr{\cdot}$ is to guarantee 
that for $u\in\anbr{B}$ and $v\in\uintd$ the componentwise ordering 
$v\le u$ implies $v\in\anbr{B}$. This immediately yields 
$\anbr{\anbr{B}} = \anbr{B}$. Consistently with \cite{Dudley:1999}, we will 
call $\anbr{B}$ the \emph{lower layer} of $B$. 
{\secondrevision%
This set class is also mentioned in the context of nonparametric regression 
\citep[cf.][and references therein]{Wright:1981}.}%
\item \label{item:rem.1.c}
The sets $\anbr{T(\At)}$ are closed if the marginal distributions $F_i$ have 
bounded domains, but not necessarily in the general case.
If, for instance, $F_1=F_2$ are 
standard normal distributions, then $\anbr{T(A_0)}= \cubr{u\in[0,1]^2: u_1+u_2\le 1}\setminus\cubr{(0,1),(1,0)}$. This example with punctured corners is quite 
prototypical. 
It is easy to show that if a sequence $u\sampn$ in $\anbr{T(\At)}$ converges 
to $u\notin\anbr{T(\At)}$, then $u$ is on the boundary of $\uintd$. 
Indeed, if $u\in(0,1)^d$, then $u\sampn\in(0,1)^d$ for sufficiently large $n$.
This allows to construct a sequence $v\sampn \to u$ such that $v\sampn \le u\sampn$ and $v\sampn \le u$ for all $n$. For any $u\sampn\in\anbr{T(A_t)}$ there exists $x\sampn\in \At$ such that 
$T(x\sampn)\ge u\sampn$. As $T\ginv$ is non-decreasing and $F_i\ginv\circ F(x_i) \le x_i$ for all $i$, we have $x\sampn \ge T\ginv(u\sampn) \ge T\ginv(v\sampn)$ and hence $T\ginv(v\sampn)\in\At$ for all $n$. Since all $F_i\ginv$ are left continuous on $(0,1)$, we obtain $T\ginv(v\sampn) \to T\ginv(u)$ and $T\ginv(u)\in\At$. As $F_i\circ F_i\ginv(u_i)\ge u_i$ for all $i$, we obtain $u\in\anbr{T(\At)}$.
Thus $u\notin\anbr{T(\At)}$ is only possible for $u\in\bnd\uintd$.
\par
By construction, $\anbr{T(\At)}$ includes all points $u\in\bnd\uintd$ 
such that $u+\eps e_i\in\anbr{T(\At)}$ for some unit vector $e_i$, $i=1,\ldots,d$, and $\eps>0$. Thus the area where the set $\anbr{T(\At)}$ does not include its boundary points is very small. 
\end{enumerate}
\end{remark}
Let us now return to the estimation of the sum distribution 
$G(t)=\Prob_F(\At)$. 
The empirical distribution of the component sum in the 
synthetic sample~\eqref{eq:047} is nothing else than the 
empirical multivariate distribution of this sample evaluated 
at the sets $\At$:
\begin{align}
\Gnast(t)
&:= \label{eq:065}
\onebyn
\sumkonen 1\cubrfl{X_1^{(R_1^{(k)}:n)}+\ldots+X_d^{(R_d^{(k)}:n)} \le t}\\
&= \nonumber
\onebyn
\sumkonen 1\cubrfl{\robrfl{X_1^{(R_1^{(k)}:n)},\ldots,X_d^{(R_d^{(k)}:n)}} \in\At}. 
\end{align}
{\secondrevision%
Analogously to~\eqref{eq:062}, $\Gnast(t)=\ProbFnast(\At)$ can be written 
in terms of the empirical copula $\Cnast$ defined in~\eqref{eq:063} 
and, as next step, in terms of the 
empirical distribution $\Cn$ of the i.i.d.\ copula sample 
$U\sampone,\ldots,U\sampn$:
\[
\Cn(u) := \onebyn\sumkonen 1 \cubr{U\sampk \in[0,u]},
\quad u\in\uintd.
\]
Let $\Cin$ denote the margins of $\Cn$, and let $\Cin\ginv$ denote the 
corresponding quantile functions.  
To avoid technicalities, we consider $\Cin\ginv$ as mappings from $[0,1]$ to 
$[0,1]$:
\begin{align*} 
\Cin\ginv(u) := \inf \cubr{v\in[0,1]: \Cin(v) \ge u},
\quad u\in[0,1].  
\end{align*}
Denote $\tau_n(x):=\robr{\Fonen(x_1),\ldots,\Fdn(x_d)}$ and 
$\Tn:=\rho_n\ginv\circ\tau_n$, where  
$\rho_n\ginv:=\robr{\Conen\ginv(x_1),\ldots,\Cdn\ginv(x_d)}$.
Then we can state the following result. 
\begin{corollary} \label{cor:5}
\begin{align}\label{eq:012}
\forall t\in\R
\quad
\Gnast(t) 
= 
\ProbCnast\robr{\anbr{\tau_n(\At)}} 
\end{align}
and, with probability $1$,
\begin{align} \label{eq:013}
\forall t\in\R
\quad
\Gnast(t) 
= 
\ProbCn\robr{\anbr{\Tn(\At)}}.
\end{align}
\end{corollary}
\begin{proof}
It is easy to see that the synthetic sample~\eqref{eq:047} can be written as 
\[
\Xtilde\sampk = \tau_n\ginv\circ\rho_n\robrfl{U\sampk}
, \quad k=1,\ldots,n,
\]
where $\tau_n\ginv(x):=\robr{\Fonen\ginv(x_1),\ldots,\Fdn\ginv(x_d)}$ 
and $\rho_n(x):=\robr{\Conen(x_1),\ldots,\Cdn(x_d)}$. 
This yields 
\begin{align} \label{eq:064}
\Gnast(t)
=
\onebyn\sumkonen 1\cubrfl{\tau_n\ginv\circ\rho_n\robrfl{U\sampk} \in \At}. 
\end{align}
According to the proof of Lemma~\ref{lem:5}, 
$\tau_n\ginv(x)\in\At$ is equivalent to $x\in\anbr{\tau_n(\At)}$.  
Hence~\eqref{eq:064} implies
\[
\Gnast(t)
=
\onebyn\sumkonen 1\cubrfl{\rho_n\robrfl{U\sampk} \in \anbr{\tau_n\At}},
\] 
which is the same as~\eqref{eq:012} because $\Cnast$ is the empirical 
distribution function of $\rho_n(U\sampone),\ldots,\rho_n(U\sampn)$. 
\par
Being continuously distributed,
$U_i\sampone,\ldots,U_i\sampn$ 
have different values $\Prob$-a.s.\ for each $i$.  
Hence the mapping $\rho_n\ginv$ is componentwise $\Prob$-a.s.\ 
strictly increasing on $\cubr{\onebyn,\ldots,1}^d$  
with probability $1$, and therefore 
\begin{align*}
1\cubrfl{\rho_n\robrfl{U\sampk} \in \anbr{\tau_n(\At)}} = 
1\cubrfl{\rho_n\ginv\circ\rho_n\robrfl{ U\sampk} \in \anbrfl{\rho_n\ginv\robr{\anbr{\tau_n(\At)}}}}
\quad
\Prob\text{-a.s.}
\end{align*}
Thus \eqref{eq:013} follows from 
$\rho_n\ginv\circ\rho_n(U\sampk) = U\sampk$
and 
$\anbr{\rho_n\ginv\robr{\anbr{\tau_n(\At)}}} = \anbr{\rho_n\ginv\circ\tau_n(\At)}$.
\end{proof}
\begin{remark}
Uniform consistency of $\Fin$ and $\Cin\ginv$ implies  
$\Tn\to T$ in $\linf(\Rd)$.
\end{remark}
}%
\section{Complexity of the problem}\label{sec:2a}
The representation~\eqref{eq:013} translates the asymptotic normality of 
$\Gnast$ into a CLT for $\Cn$ uniformly on the random set sequence 
$\robr{\anbr{\Tn(\At)}: n\in \N}$. 
The canonical way to prove results of this kind is to establish a uniform CLT 
on the set class $\Tcald$ of all possible $\anbr{\Tn(\At)}$ and $\anbr{T(\At)}$. 
This is the natural set class to work with if $\Fi$ are unknown and estimated 
empirically. 
The index $d$ in the notation $\Tcald$ highlights the dimension. 
Since $F_i$ need not be continuous, the set of all $T$ includes all $\Tn$, 
so that $\Tcald$ is simply the collection of all possible $\anbr{T(\At)}$.
{\secondrevision
Furthermore, if each unknown margin $F_i$ has a positive density on entire 
$\R$, then the resulting empirical distributions $\Fin$ can take 
any value in the class of all possible stair  
functions on $\R$ with steps of size $1/n$ going from $0$ to $1$. 
In this case the class of all possible $\anbr{\Tn(\At)}$ is dense 
(w.r.t.\ Hausdorff metric) in the class of all possible $T(\At)$. 
Thus, even though there is only one limit transformation $T$ that really 
matters to us, considering the class of all possible 
$\anbr{T(\At)}$ does not add more complexity to the problem.}%
\par
It is also easy to see that even pointwise asymptotic normality of $\Gnast(t)$ 
in some $t=t_0$ would require a uniform CLT on $\Tcald$. 
Shifting the unknown margins $F_i$, one can easily generate all possible sets $T(\At)$  from a single $A_{t_0}$. 
Thus the complexity of the problem is the same for the uniform 
and for the pointwise asymptotic normality of $\Gnast$.
\par
There are various functional CLTs for empirical copulas 
\citep[cf.][and references therein]{Rueschendorf:1976,Deheuvels:1979,Fermanian/Radulovic/Wegkamp:2004,Segers:2012}.
However, empirical copula estimates are empirical measures evaluated on 
the set class $\Rcal_d$ of rectangle cells 
$(-\infty,a]:=(-\infty,a_1]\times\ldots\times(-\infty,a_d]$ for $a\in\Rd$. 
This set class is simple enough to be \emph{universally Donsker}. 
A set class $\Ccal$ is called $\Prob$-\emph{Donsker} if the empirical measure 
$\Pbb_n(B):=\frac{1}{n}\sumkonen 1_B(Y\sampk)$ of an i.i.d.\ sample $Y\sampone,\ldots,Y\sampn \sim \Prob$ 
satisfies
\begin{align} \label{eq:055}
\sqrt{n}\robrfl{\Pbb_n(B) -\Prob(B)} \weakconv \Gbb_\Prob(B)
\end{align}
as a mapping in $\linf(\Ccal)$, where $\Gbb_\Prob$ is the so-called 
Brownian bridge \enquote{with time} $P$. 
That is, $\Gbb_\Prob$ is a centred Gaussian process  
with index $B\in\Ccal$ and covariance structure 
\[
\cov(\Gbb_\Prob(A),\Gbb_\Prob(B)) = \Prob(A \cap B) - \Prob(A)\Prob(B).
\]
The Donsker property of $\Ccal$ is called \emph{universal} if  
it holds for any probability measure $\Prob$ on the sample space. 
\par
The symbol $\weakconv$ in~\eqref{eq:055} refers to the extended notion of 
weak convergence for non-measurable mappings in $\linf(\Ccal)$ 
as used in~\cite{van_der_Vaart/Wellner:1996}. 
See Remark~\ref{rem:6} for further details. 
\par
Sufficient conditions for a set class to be Donsker can be obtained from the 
\emph{entropy} of this set class.  
Entropy conditions can be formulated in terms of \emph{covering numbers}
or \emph{bracketing numbers} 
\citep[cf.][Sections 2.1 and 2.2]{van_der_Vaart/Wellner:1996}. 
%
Entropy bounds that do not depend on the underlying probability measure
are called \emph{uniform}. 
The most common sufficient criterion for uniform entropy bounds guaranteeing 
that a set class is universally Donsker is the  
%
\emph{\VC\ (VC)} property. 
A set class $\Ccal$ is VC if it does not \emph{shatter} any $n$-point set  
$\cubr{x\sampone,\ldots,x\sampn}$ for sufficiently large $n$. 
The set $\cubr{x\sampone,\ldots,x\sampn}$ is shattered by $\Ccal$ if every 
subset $A\subset\cubr{x\sampone,\ldots,x\sampn}$ can be obtained as $A=B\cap\cubr{x\sampone,\ldots,x\sampn}$ with some $B\in\Ccal$. The smallest $n$ such that no $n$-point set is shattered by $\Ccal$ is called \emph{VC-index} of $\Ccal$. 
\par
It is well known that the set class $\Rcal_d$ is VC with index $d+1$ 
\citep[cf.][Example 2.6.1]{van_der_Vaart/Wellner:1996}.
This yields asymptotic normality of $\Cn(u)$ uniformly in $u\in\uintd$.  
The asymptotic normality of empirical copulas follows then by 
functional Delta method. 
These functional CLTs allow to prove asymptotic normality for the estimators of the multivariate distribution function $F(x)$ that are derived from $\Cn$ or $\Cnast$.  
\par
Unfortunately, the problem for $\Gnast$ is much more difficult.
As shown above, a functional CLT for $\Gnast$ is closely related to a 
uniform CLT for $\Cn$ on the set class $\Tcald$. 
The complexity of $\Tcald$ is much higher than that of $\Rcal_d$. 
Lemma~\ref{lem:3} stated below implies that $\Tcald$ is not VC, 
and Remark~\ref{rem:2}(\ref{item:rem.2.b}) 
shows that the complexity of this set class is even so high that 
a uniform CLT on $\Tcal_d$ does not hold. 
\par
Let $\Hcald$ denote the collection of all lower layers in $\uintd$:
\[
\Hcald:=\cubrfl{B\subset\uintd : \anbr{B} = B},
\]
and denote $\Hcaldbar := \cubr{\overline{B} : B \in\Hcald}$. Analogously, 
denote $\Tcaldbar := \cubr{\overline{B} : B \in\Tcald}$. According to 
Remark~\ref{rem:1}(\ref{item:rem.1.c}), $B\in\Tcald$ implies that $\overline{B}\setminus B \subset\bnd\uintd$, 
which is a $\ProbC$-null set for any copula $C$. Thus, for empirical 
processes constructed from copula samples, uniform convergence on $\Tcald$ 
is equivalent to uniform convergence on $\Tcaldbar$.
\par 
It is obvious that $\Tcaldbar\subset\Hcaldbar$. 
The following result shows that for $d=2$ these set classes are almost 
identical. 
\begin{lemma} \label{lem:3}
If $B\in\Hcaltwobar$, then $B\cup\Lambda_2\in\Tcaltwobar$,  
where $\Lambda_2:=\cubr{u\in[0,1]^2: u_1u_2=0}$ is the union of 
lower faces of $[0,1]^2$.  
\end{lemma}
\begin{proof}
Denote $B':=B\cup\Lambda_2$. 
It suffices to find probability distribution functions $F_1,F_2$ such that 
$\anbr{T(A_0)}=B'$ with $T(x)=(F_1(x_1),F_2(x_2))$. 
Denote 
\[
F_1(t):=t1_{[0,1)}(t) + 1_{[1,\infty)}(t)
\]
and
\[
F_2(t):=
\begin{cases}
0 & \text{if }t<-1\\
\sup\cubr{s\in[0,1]: (-t,s)\in B'} & \text{if }t\in[-1,0)\\
1 & \text{if }t\ge0
\end{cases}.
\]
Since $B'\in\Hcaltwobar$, the function $F_2$ is non-decreasing in $t$.
Indeed, if $(-t,s)\in B'$, then $(-t-\delta,s)\in B'$ for any $\delta\in(0,t]$, 
and hence $F_2(t+\delta)\ge F_2(t)$. 
The maximal value of $F_2$ is $1$, and $F_2$ is right continuous because $B'$ is 
closed. Thus $F_2$ is a probability distribution function. 
\par
As $F_1$ and $F_2$ are non-decreasing, we have 
\[
\anbr{T(A_0)} = \anbr{T(\bnd A_0)}
\]
where $\bnd A_0:= \cubr{(-t,t): t\in \R}$ is the boundary of $A_0$. Now observe 
that 
\begin{align*}
T(\bnd A_0) 
=
\cubrfl{(F_1(t), F_2(-t)) : t\in\R} 
=
\cubrfl{(t,F_2(-t)) : t\in[0,1]}.  
\end{align*}
Hence $\anbr{T(A_0)}$ is the area enclosed between the zero line and 
the graph of $F_2(-t)$ for $t\in[0,1]$. This is precisely $B'$. 
\end{proof}
\begin{remark} \label{rem:2}
\begin{enumerate}[(a)]
\item
Since $\Lambda_2$ is a $\ProbC$-null set for any copula $C$, 
the modification of $B$ into $B\cup\Lambda_2$ in 
Lemma~\ref{lem:3} has no influence on the uniform convergence of empirical 
processes obtained from copula samples. 
\item \label{item:rem.2.a}
The set classes $\Hcaldbar$ and $\Hcald$ are not VC. 
For instance, they shatter all sets 
$\cubr{u\in\cubr{0,\frac{1}{n},\ldots,1}^d: u_1+\ldots+u_d=1}$ for $n\in\N$. 
Any subset $B$ of this hyperplane in $\cubr{0,\frac{1}{n},\ldots,1}^d$ 
can be picked out 
by $\anbr{B}\in\Hcaldbar\subset\Hcald$. 
Similar arguments apply to the modified set class 
$\cubr{B\cup\Lambda_d: B\in\Hcaldbar}$. 
Hence Lemma~\ref{lem:3} implies that $\Tcaltwobar$ is 
not VC. This rules out the canonical usage of VC criteria in 
convergence proofs for $\Gnast$.
\item \label{item:rem.2.b}
The problem is even more difficult, and also more remarkable. 
In fact, the set classes $\Hcald$ and $\Hcaldbar$ for $d\ge2$ are not Donsker 
with respect to the Lebesgue measure on $\uintd$ \citep[cf.][Theorems 8.3.2, 12.4.1, and 12.4.2]{Dudley:1999}. 
Thus Lemma~\ref{lem:3} implies that a uniform CLT on $\Tcaltwobar$ does not hold in the most basic case, when $C$ is the independence copula.
Therefore one cannot prove asymptotic normality of $\Gnast$ via uniform CLT on $\Tcal_2$, and precise asymptotic variance of $\Gnast$ also seems out of reach. 
\item\label{item:rem.2.c}
The complexity issues are not specific to the estimator $\Gnast$ obtained by
plugging empirical margins $\Fin$ into the rank based empirical copula $\Cnast$. 
They also affect models generated by plugging $\Fin$ directly into the \enquote{exact} copula $C$. Top-down simulation of such models means marginal 
transformation of copula samples $U\sampone,\ldots,U\sampn\sim C$ by $\Fin\ginv$. The resulting estimate of the component sum distribution $G(t)$ can be written as $\ProbCn(\anbr{\tau_n(A_t)})$ where 
$\tau_n(x):=(\Fonen(x_1),\ldots,\Fdn(x_1))$. 
The sets $\anbr{\tau_n(\At)}$ feature the same stair shape as the sets 
$\anbr{\Tn(\At)}$. Asymptotic normality of $\ProbCn(\anbr{\tau_n(\At)})$ 
leads us again to a uniform CLT on $\Tcaldbar$. 
\item
It is not yet clear whether Lemma~\ref{lem:3} can be extended to 
$B\in\Hcaldbar$ for $d>2$.  
However, the complexity of $\Hcaldbar$ can only increase for greater $d$. 
It is easy to embed $\Hcaltwobar$ in $\Hcaldbar$ for $d>2$ by identifying 
$\Hcaltwobar$ with the following subclass:
\[
\Hcalbar_{d|1,2}\cubrfl{B\in\Hcaldbar: B= B'\times[0,1]^{d-2}, B'\in\Hcaltwobar}.
\]
Setting $F_i(t):=1_{[0,\infty)}(t)$ for $i=3,\ldots,d$ in the proof of 
Lemma~\ref{lem:3}, the result obtained there can be 
extended to $B\cup\Lambda_d\in\Tcaldbar$ for $B\in\Hcalbar_{d|1,2}$.
This allows to extend the conclusions in~(\ref{item:rem.2.a},\ref{item:rem.2.b},\ref{item:rem.2.c}) to all dimensions $d>2$. 
\end{enumerate}
\end{remark}
\par %
\begin{remark}
\begin{enumerate}[(a)]
\item
The results of this section can be summarized as follows: 
The true target set class $\HcalT:=\cubr{\anbr{T(\At)} : t\in\R}$ is simple 
(it will be shown in Remark~\ref{rem:8} that $\HcalT$ is VC with index 2), 
but unknown. 
Replacing these unknown margins by the empirical ones, we obtain random 
elements of the set class $\Tcald$, which is too complex for a uniform CLT. 
This is the reason why the convergence proofs presented below sacrifice 
precise asymptotic variance.  
The resulting loss of precision can be considered as the price one is forced to be pay for using 
empirical margins $\Fin$ instead of the true ones. 
\item
Similar issues can also arise when the margins are known, but the copula is not. 
The implicit use of transformations $\rho_n$ in $\Cnast$ (cf.\ proof of Corollary~\ref{cor:5}) 
entails deformations of target sets $B\in\HcalT$ that are very similar to the ones caused by $\Tn$ or $\tau_n$. 
Thus, depending on the application, loss of precision may also be caused by the use of an empirical copula. 
In particular, a uniform CLT for $\Cnast$ on the set class $\HcalT$ is still an open problem, 
and the foregoing results suggest a plausible explanation why this problem is so hard. 
\item
A deeper reason behind these complexity issues 
is the typical
shape of the target sets $B\in\HcalT$. 
If the margins or the copula are estimated empirically,  
corresponding random transformations of $B$ 
($\tau_n$ and $\rho_n\ginv$ in Corollary~\ref{cor:5})  
significantly increase the complexity of the problem. 
Depending on the application, the resulting loss of  
precision can be attributed to empirical margins, 
empirical copulas, or both. 
\end{enumerate}
\end{remark}
\section{Convergence results}\label{sec:3}
The major problem studied in this paper 
is the uniform convergence of 
{\secondrevision%
the Iman--Conover estimator 
$\Gnast$ introduced in~\eqref{eq:065}.  }%
Strong consistency in $\linf(\R)$ is established in Theorem~\ref{thm:1}.   
A sufficient condition for the convergence rate to be $n^{-1/2}$ is  
given in Theorem~\ref{thm:2}. 
These results are stated below and followed by some corollaries, remarks, and 
auxiliary results needed in the proofs. 
The technical proofs of the auxiliary results are provided in Section~\ref{subsec:3.2}.
\par
\begin{theorem} \label{thm:1}
{\secondrevision%
Let $G(t):=\Prob(X_1+\ldots+X_d \le t)$ for $X\sim F$ 
as defined in~\eqref{eq:052}, 
and let $\Gnast$ be the Iman--Conover estimator of $G$ 
introduced in~\eqref{eq:065}.  }%
Assume that  the functions $\Fin$ defined in~\eqref{eq:014} satisfy 
\begin{align} \label{eq:028}
\norm{\Fin - \Fi}_\infty \to 0 \ \Prob\text{-a.s.}
,\quad i=1,\ldots,d.
\end{align}
If the copula $C$ of $F$ is Lebesgue absolutely continuous, 
then $\norm{\Gnast - G}_\infty \to 0$ $\Prob$-a.s.
\end{theorem}
\par
As discussed in Remark~\ref{rem:2}(\ref{item:rem.2.b}), a CLT for $\Gnast$ seems out of reach, 
so that the convergence rate is established as an $O_\Prob(n^{-1/2})$ bound. 
This notation is related to tightness:   
$Z_n=O_\Prob(1)$ means that $Z_n$ is tight, and 
$Z_n=O_\Prob(a_n)$ is equivalent to $a_n^{-1} Z_n = O_\Prob(1)$.
In particular, if $\abs{Z_n} \le \abs{Y_n}$ and $Y_n\weakconv Y$, then 
$Z_n=O_\Prob(1)$. 
%
%
\par 
The regularity assumptions also need some additional notation. 
In the following, let $\Bt$ denote 
the \enquote{upper} boundary 
of $\anbr{T(\At)}$:
\begin{align*} 
\Bt:=\cubrfl{x \in \overline{\anbr{T(\At)}} : \forall \eps>0 \quad x+\robr{\eps,\ldots,\eps} \notin \overline{\anbr{T(\At)}}}
\quad
t\in\R, 
\end{align*}
and let $\Udelta(\Bt)$ denote the closed $\delta$-neighbourhood of $\Bt$ in Euclidean 
distance:
\begin{align} \label{eq:017}
\Udelta(\Bt) := \cubrfl{u \in \uintd: \absfl{u-v} \le \delta \text{ for some } v\in\Bt}. 
\end{align}
One of the regularity assumptions in Theorem~\ref{thm:2} 
specifies the probability mass that 
the copula $C$ assigns to $\Udelta(\Bt)$. 
The other one involves the Lebesgue density $c$ of $C$. 
For $\eps\in(0,1/2)$, we denote 
\[
K(\eps) := \esssup \cubrfl{c(u) : u\in\uintepsd}.  
\]
The growth of $K(\eps)$ for $\eps\to 0$ specifies the behaviour of $c$ 
near the boundary of $\uintd$. 
\par
\begin{theorem} \label{thm:2}
{\secondrevision%
Let $G(t):=\Prob(X_1+\ldots+X_d \le t)$ for $X\sim F$ 
as defined in~\eqref{eq:052}, 
and let $\Gnast$ be the Iman--Conover estimator of $G$ 
introduced in~\eqref{eq:065}.  }%
Assume that the functions $\Fin$ defined in~\eqref{eq:014} satisfy 
\begin{align} \label{eq:040}
\norm{\Fin - \Fi}_\infty  = O_\Prob(n^{-1/2})
,\quad i=1,\ldots,d,
\end{align}
and that the copula $C$ of $F$ is absolutely continuous and satisfies
\begin{align} \label{eq:024}
\sup_{t \in \R} \ProbC(\Udelta(\Bt)) = O(\delta)
\end{align}
for $\delta \to 0$ and $\Udelta(\Bt)$ defined in~\eqref{eq:017}. 
Further, assume that 
\begin{align} \label{eq:025}
\int_0^{1/2} \sqrt{ \log K \robrfl{\eps^2}} \dm \eps < \infty.
\end{align}
Then 
$\norm{\Gnast - G}_\infty = O_\Prob(n^{-1/2})$. 
\end{theorem}
\begin{remark}
The shapes of the sets $\Bt$ strongly depend on the marginal 
distributions $F_i$. 
Two examples of set families $\cubr{\Bt : t \in \R}$ are given in 
Figure~\ref{fig:3}. 
The verification of~\eqref{eq:024} for some copula families is 
discussed in Section~\ref{sec:4}. The examples presented there 
suggest that this condition is non-trivial, 
and that it depends on the interplay between the copula $C$ 
and the true, unknown margins $\Fi$. 
\par
\begin{figure} 
\centering
\includegraphics[width=\textwidth]{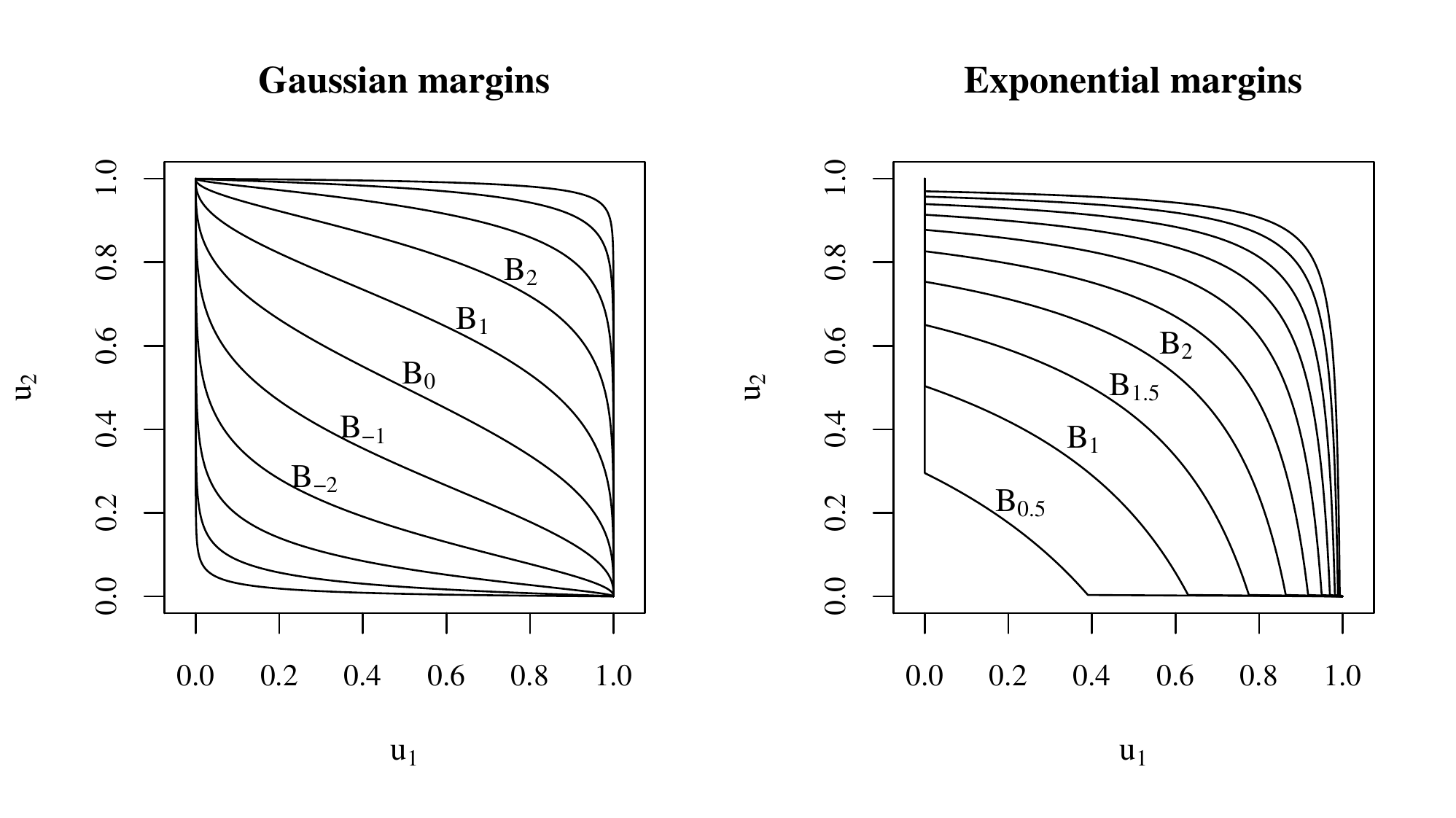}
\caption{Set families $\cubr{B_t:t\in\R}$.  
Left hand side:  $X_1\sim\Ncal(0,1)$ and $X_2\sim\Ncal(0, 1/4)$.
Right hand side: $X_1\sim\ExpDistr(1)$ and  $X_2\sim\ExpDistr(0.7)$.} 
\label{fig:3}
\end{figure}
\end{remark}
\par
We proceed with an auxiliary result that gives us an upper bound for 
the volume of  $\Udelta(\Bt)$. 
It follows from the componentwise monotonicity of the transformation $T$. 
The general idea behind this result is that the \enquote{surface area} 
of $\Bt$ is
bounded by the sum of its $d$ projections on the lower faces of the 
unit square $\uintd$, and the \enquote{thickness} of $\Udelta(\Bt)$
is roughly $2\delta$. The proof is given in Section~\ref{subsec:3.2}.
\begin{lemma}
\label{lem:1}
Let $\lambda$ denote the Lebesgue measure on $\uintd$. Then 
$\lambda(\Udelta(\Bt)) \le 2d\delta$. 
\end{lemma}
\par
The next lemma provides the Glivenko--Cantelli and Donsker properties for two 
set classes involved in the proofs of Theorems~\ref{thm:1} and~\ref{thm:2}. 
The Donsker property is defined in~\eqref{eq:055}.  
A set class $\Ccal$ is called \emph{$\Prob$-Glivenko--Cantelli}  
if the empirical measure 
$\Pbb_n(B):=\frac{1}{n}\sumkonen 1_B(Y\sampk)$ 
of an i.i.d.\ sample $Y\sampone,\ldots,Y\sampn \sim \Prob$ 
satisfies
\begin{align} \label{eq:016}
\norm{\Pbb_n - \Prob}_\Ccal
:=
\sup_{B\in\Ccal}\abs{\Pbb_n(B) - \Prob(B)} \asconv 0 
\end{align}
This notation emphasizes that the convergence  
also depends on the true distribution that is sampled to construct $\Pbb_n$. 
\begin{remark} \label{rem:6}
One technical aspect of~\eqref{eq:016} and~\eqref{eq:055} 
needs an additional comment. 
These statements regard 
$\Pbb_n$ and $\Prob$ 
as mappings from the probability space $(\Omega,\Acal,\Prob)$ 
to $\linf(\Ccal)$.
However, $\Pbb_n$ need not be measurable with respect to 
the Borel $\sigma$-field on $\linf(\Ccal)$ 
\citep[cf.][Chapter 18]{Billingsley:1968}. 
This issue can be solved by extended versions of almost sure 
and weak convergence as presented in   
\citet{van_der_Vaart/Wellner:1996}. 
In the following, $\asconv$ and $\weakconv$ are understood 
according to that monograph. 
In case of measurability these extended notions 
coincide with the standard ones.  
\end{remark}
\begin{lemma} \label{lem:4}
\begin{enumerate}[(a)]
\item \label{item:lem.4.a}
The set class 
\[
\HcalT:=\cubr{\anbr{T(\At)}:t\in\R}
\] 
for a fixed $T(x)=\robr{F_1(x_1),\ldots,F_d(x_d)}$  %
is universally Glivenko--Cantelli and Donsker.
\item \label{item:lem.4.b}
If $C$ is Lebesgue absolutely continuous, then the set class  
\begin{align*} 
\Dcaldeltazer:=\cubrfl{\Udelta(\Bt):\delta\in[0,\delta_0], t\in \R}
\end{align*}
is $\ProbC$-Glivenko--Cantelli for 
any $\delta_0>0$.
\item \label{item:lem.4.c}
If $C$ is Lebesgue absolutely continuous and satisfies~\eqref{eq:025}, then 
$\Dcaldeltazer$ is $\ProbC$-Donsker for any $\delta_0>0$.
\end{enumerate}
\end{lemma}
The proof of this auxiliary result is given in Section~\ref{subsec:3.2}. 
Now we proceed with the proofs of Theorems~\ref{thm:1} and~\ref{thm:2}. 
\par
\begin{proof}[Proof of Theorem~\ref{thm:1}]
For the sake of simplicity, we will write $\mu\anbr{B}$ 
instead of $\mu(\anbr{B})$ for any measure $\mu$ on $\uintd$.
According to~\eqref{eq:013}, we have to show that 
$\ProbCn\anbr{\Tn(\At)} \to \ProbC\anbr{T(\At)}$ uniformly in $t\in\R$. 
It is easy to see that
\begin{align}
\mylefteqn
\absfl{\ProbCn\anbr{\Tn(\At)} - \ProbC\anbr{T(\At)}} \nonumber\\
&\le \label{eq:015}
\ProbCn \robrfl{\anbr{\Tn(\At)} \symdiff \anbr{T(\At)}} 
+
\absfl{\ProbCn\anbr{T(\At)} - \ProbC\anbr{T(\At)}},
\end{align}
where $\symdiff$ denotes the symmetric difference: 
$A \symdiff B:= (A\setminus B) \cup (B\setminus A)$. 
According to Lemma~\ref{lem:4}(\ref{item:lem.4.a}), the set class 
$\HcalT$ is $\ProbC$-Glivenko--Cantelli. 
Hence the second term in~\eqref{eq:015} converges to $0$ $\Prob$-a.s.\ 
uniformly in $t\in\R$.
%
%
\par
Now consider the first term in~\eqref{eq:015} and denote 
\begin{align} \label{eq:059}
Y_n:=\norm{\Tn (x) - T(x)}_\infty.
\end{align}
As the transformations $\Tn$ and $T$ are componentwise non-decreasing, 
$Y_n$ is a measurable random variable. 
Furthermore, symmetry arguments give us 
$\norm{\Cin\ginv - \id_{[0,1]}}_\infty = \norm{\Cin - \id_{[0,1]}}_\infty$, where
$\id_{[0,1]}(u):=u$ for $u\in[0,1]$. 
Hence~\eqref{eq:028} and the classic Glivenko-Cantelli theorem for $\Cin$ 
yield $\Tn\asconv T$ in $\linf(\Rd)$. 
This implies that $Y_n\asconv 0$. 
\par
It is also easy to see that
\begin{align} \label{eq:046}
\anbr{\Tn(\At)} \symdiff \anbr{T(\At)} 
\subset 
U_{Y_n} (\Bt),
\end{align}
where $\Udelta(\Bt)$ is the set introduced in \eqref{eq:017}.
%
Moreover, for any $\delta>0$ we have    
\begin{align} \label{eq:053}
\ProbCn\robrfl{\Udelta(\Bt)} 
\le 
\ProbC\robrfl{\Udelta(\Bt)} + 
\absfl{\ProbCn (\Udelta(\Bt)) - \ProbC (\Udelta(\Bt))}.
\end{align}
As $Y_n\asconv 0$, 
it suffices to show that for $\delta\to0$ both terms 
on the right hand side of~\eqref{eq:053}
vanish with probability $1$ uniformly in $t$. 
In particular, for the second term it suffices to show that 
for some $\delta_0>0$
\[
\lim_{n\toinf}
\sup_{t\in\R,\delta\in[0,\delta_0]} 
\absfl{\ProbCn (\Udelta(\Bt)) - \ProbC (\Udelta(\Bt))}
=0 
\quad \Prob\text{-a.s.}
\]
This follows from Lemma~\ref{lem:4}(\ref{item:lem.4.b}). 
\par
The first term on the right hand side of~\eqref{eq:053} vanishes 
due to the absolute continuity of the copula $C$. 
Indeed, let $\eps>0$. Since the density $c$ of $C$ is non-negative 
and $\int c(u)\,\dm\lambda(u)=1$, 
there exists $M>0$ such that $\int_{\cubr{c>M}} c(u)\,\dm\lambda(u) < \eps/2$. 
Then, for $\delta \le \eps/(4dM)$, Lemma~\ref{lem:1} yields
\begin{align*}
\ProbC\robrfl{\Udelta(\Bt)} 
& \le 
\ProbCn \robrfl{\Udelta(\Bt) \cap \cubr{c\le M}} + \frac{\eps}{2}\\
& \le M \lambda(\Udelta(\Bt)) + \frac{\eps}{2}\\
& \le M2d\delta + \frac{\eps}{2}\ = \eps.
\end{align*}
That is, $\ProbC(\Udelta(\Bt))\to 0$ for $\delta\to 0$
\end{proof}
\par
\begin{proof}[Proof of Theorem~\ref{thm:2}]
According to~\eqref{eq:015} and \eqref{eq:046}, we have that   
\begin{align}
\mylefteqn
\sqrtn \absfl{\Gnast(t) - G(t)} \nonumber \\
&\le \label{eq:020}
\sqrtn
\ProbCn \robrfl{U_{Y_n}(\Bt)} +
\sqrtn
\absfl{ \ProbCn \anbr{T(\At)} - \ProbC \anbr{T(\At)}}.
\end{align}
The second term in~\eqref{eq:020} is $O_\Prob(1)$ uniformly in $t\in\R$ 
due to Lemma~\ref{lem:4}(\ref{item:lem.4.a}).
\par
Now consider the first term in~\eqref{eq:020} and observe that  
\begin{align}
\nonumber
\mylefteqn 
\sqrtn\ProbCn \robrfl{ U_{Y_n}(\Bt)}\\
&= 
\label{eq:022}
\sqrtn \robrfl{ \ProbCn(U_{Y_n}(\Bt)) - \ProbC(U_{Y_n}(\Bt)) }
+
\sqrtn \ProbC \robrfl{U_{Y_n}(\Bt)}.
\end{align}
Applying the classic Donsker Theorem 
\citep[cf.][Theorem 2.5.7]{van_der_Vaart/Wellner:1996} 
to $\Cin$, we obtain that 
$\norm{\Cin\ginv - \id_{[0,1]}}_\infty = \norm{\Cin - \id_{[0,1]}}_\infty$ 
is $O_\Prob(n^{-1/2})$. 
Hence assumption~\eqref{eq:040} yields $Y_n=O_\Prob(n^{-1/2})$, 
and assumption~\eqref{eq:024} 
implies that the second term in~\eqref{eq:022} is $O_\Prob(1)$. 
\par
Let $Z_n$ denote the the first term in \eqref{eq:022}. 
As $Y_n = o_\Prob(1)$, we have 
\begin{align*} 
Z_n 
= 1\cubr{Y_n\le \delta_0}\sqrtn \robrfl{ \ProbCn(U_{Y_n}(\Bt)) - \ProbC(U_{Y_n}(\Bt)) }
+ o_\Prob(1). 
\end{align*}
for any $\delta_0>0$. 
Now assumption~\eqref{eq:025} and Lemma~\ref{lem:4}(\ref{item:lem.4.c}) imply 
that $Z_n$ is weakly convergent, and hence $O_\Prob(1)$. 
\end{proof}
\par
The following corollary allows to replace the empirical marginal distributions 
$\Fin$ in Theorems~\ref{thm:1} and~\ref{thm:2} by any other consistent 
approximations of the true, unknown margins $\Fi$.
\begin{corollary} \label{cor:4} 
Let $\Fin$, $i=1,\ldots,d$, $n\in\N$, be arbitrary distribution 
functions on $\R$, and let $\Gnast(t) := \Prob_{\Fnast}(\At)$ with $\Fnast(x):=\Cnast(\Fonen(x_1),\ldots,\Fdn(x_d))$.
\begin{enumerate}[(a)]
\item \label{item:cor.4.a}
If $\Fin$ satisfy~\eqref{eq:028} and $C$ is absolutely continuous,
then 
$\norm{\Gnast - G}_\infty \to 0$ $\Prob$-a.s.\ 
\item \label{item:cor.4.b}
If $\Fin$ satisfy~\eqref{eq:040} 
and $C$ satisfies~\eqref{eq:024} and \eqref{eq:025}, then 
$\norm{\Gnast - G}_\infty = O_\Prob(n^{-1/2})$. 
\end{enumerate} 
\end{corollary}
\begin{proof}
Part~(\ref{item:cor.4.a}). 
For each $\Fin$ there is an approximation 
$\Ftildein:\R\to\cubr{0,\frac{1}{n},\ldots,1}$  
that minimizes $\norm{\Fin - \Ftildein}_\infty$. 
It is obvious that $\norm{\Fin - \Ftildein}_\infty \le 1/n$. 
Hence the estimator  
$\Gtilden(t):=\Cnast(\anbr{\Ttilden(A_t)})$
with $\Ttilden(A) := \robr{\Conen\ginv\circ\Ftilde_{1,n}(x_1),\ldots,\Cdn\ginv\Ftilde_{d,n}(x_d)}$ 
satisfies the assumptions of Theorem~\ref{thm:1}, and therefore  
$\norm{\Gtilden-G}_\infty \asconv 0$.
Moreover, 
\begin{align}
\abs{\Gnast(t) - \Gtilden(t)} 
&\le\nonumber
\ProbCn\robr{\anbr{\Ttilden(A_t)} \symdiff \anbr{\Tn(A_t)}}\\
&\le\label{eq:032}
\ProbCn\robr{U_{\norm{\Ttilden - \Tn}_\infty +\norm{\Tn-T}_\infty}(\Bt)}.
\end{align}
As $\norm{\Ttilden-\Tn}_\infty\asconv 0$ and $\norm{\Tn-T}_\infty\asconv 0$,  
the term \eqref{eq:032}  vanishes with probability $1$ uniformly in $t\in\R$ 
analogously to the first term in~\eqref{eq:053}.
This yields $\norm{\Gnast-G}_\infty\asconv 0$. 
\par
Part~(\ref{item:cor.4.b}). 
If $\Fin$ satisfy~\eqref{eq:040}, then so do $\Ftilde_{i,n}$.  
Hence Theorem~\ref{thm:2} yields $\norm{\Gtilden - G}_\infty =O_\Prob(n^{-1/2})$.
Furthermore, \eqref{eq:032} implies that  
\begin{align} \label{eq:034}
\abs{\Gtilden(t) - \Gnast(t)} \le \abs{\ProbCn (U_{\Ytilden} (\Bt)) - \ProbC(U_{\Ytilden} (\Bt))} + \ProbC(U_{\Ytilden} (\Bt))
\end{align}
for $\Ytilden:=\norm{\Ttilden-T}_\infty + \norm{\Tn-T}_\infty$. 
As $\Ytilden=O_\Prob(n^{-1/2})$, assumption~\eqref{eq:024} implies that 
$\sup_{t\in\R}\ProbC(U_{\Ytilden}(\Bt)) = O_\Prob(n^{-1/2})$. 
The first term on the right hand side of \eqref{eq:034} is $O_\Prob(n^{-1/2})$ uniformly in $t\in\R$ due to Lemma~\ref{lem:4}(\ref{item:lem.4.c}). 
\end{proof}
\par
\begin{remark} \label{rem:9}
\begin{enumerate}[(a)]
\item
Compared to $\Gnast$, the multivariate distribution function 
$\Fnast(x)$ obtained by plugging $\Fin$ into $\Cnast$ is much easier 
to handle. 
The deeper reason here is that $\Fnast(x)$
can be written as the empirical measure 
$\ProbCn$ indexed with random elements of the rectangle set class $\Rcal_d$. 
In particular, if $\Fin$ are defined according to~\eqref{eq:014},  
then, analogously to~\eqref{eq:013}, we have 
\begin{align}
\Fnast(x)
&=\nonumber
\Cnast(\Fonen(x_1),\ldots,\Fdn(x_d)) \\
&=\label{eq:060}
\Cn(\Conen\ginv\circ\Fonen(x_1),\ldots,\Cdn\ginv\circ\Fdn(x_d))
=\ProbCn(\anbr{\Tn(x)}). 
\end{align}
As mentioned above, $\Rcal_d$ is VC, and hence universally Donsker and 
Glivenko--Cantelli. Thus, due to $\anbr{\Tn(x)}\in\Rcal_d$, we can apply standard results to $\Fnast$. 
\item
To prove strong consistency of $\Fnast$, recall that 
any copula is a Lipschitz function with Lipschitz constant $1$ 
\citep[cf.][Theorem 2.2.4]{Nelsen:2006}. 
Therefore \eqref{eq:060} yields 
\begin{align}\label{eq:031}
\norm{\Fnast - F}_\infty \le \norm{\Cn-C}_\infty 
+
Y_n.
\end{align}
As mentioned below~\eqref{eq:059}, assumption~\eqref{eq:028} implies that 
$Y_n\asconv 0$. 
Hence $\norm{\Fnast - F}_\infty\to 0$ $\Prob$-a.s.\ due to 
the classic Glivenko--Cantelli Theorem for empirical distribution functions.  
The extension to general $\Fin$ is analogous to 
Corollary~\ref{cor:4}(\ref{item:cor.4.a}).
\item
In the proof of Theorem~\ref{thm:2} it is shown that 
assumption~\eqref{eq:040} entails $Y_n=O_\Prob(n^{-1/2})$. 
Hence the  $O_\Prob(n^{-1/2})$
convergence rate for $\Fnast$   
follows from~\eqref{eq:031} and the classic Donsker Theorem for 
empirical distribution functions. 
The extension to general $\Fin$ is analogous to 
Corollary~\ref{cor:4}(\ref{item:cor.4.b}). 
\item
If $\Fin$ satisfy a functional CLT, then the functional Delta method 
yields a functional CLT for $\Fnast$, with precise asymptotic variance -- 
see \citet[Lemma 3.9.28]{van_der_Vaart/Wellner:1996} 
and \citet{Segers:2012} for further details.
Unfortunately, this does not imply a functional CLT for $\Gnast$, as $\Gnast$ 
is obtained by indexing $\ProbCn$ with a totally different set class. 
%
\end{enumerate}
\end{remark}

\par
\begin{remark} \label{rem:7}
Theorems~\ref{thm:1} and~\ref{thm:2}, along with all their extensions and 
corollaries, also apply to multivariate 
models generated by plugging empirical margins $\Fin$ directly into 
the copula $C$. According to Remark~\ref{rem:2}(\ref{item:rem.2.c}), 
the resulting estimator of $G(t)$ can be written as $\ProbCn(\anbr{\tau_n(\At)})$ 
with $\tau_n(x)=\robr{\Fonen(x_1),\ldots,\Fdn(x_d)}$.
Since $\Gnast=\ProbCn(\anbr{\Tn(\At)})$, extension of convergence results 
to $\ProbCn(\anbr{\tau_n(\At)})$ is straightforward.
A closer look at the proof of Corollary~\ref{cor:5} suggests that same is true for the convergence 
of $\ProbCn(\anbr{\rho_n\ginv\circ T(\At)})$ uniformly in $t\in\R$, 
where $T(x)=\robr{F_1(x_1),\ldots,F_d(x_d)}$ and $\rho_n(x)=\robr{\Conen(x_1),\ldots,\Cdn(x_d)}$. 
This setting corresponds to the combination of exact margins with the 
empirical copula $\Cnast$. 
\end{remark}
The final result in this section generalizes  
all foregoing results
to a broader class of aggregation functions. 
Revising the proofs above, it is easy to see that the only property of 
the component sum used there is that it is componentwise non-decreasing. 
This immediately yields the following extension. 
\begin{corollary} \label{cor:3}
Let a function $\Psi:\Rd\to\R$ satisfy 
\[
\Psi(x) \le \Psi(y)
\quad
\text{if } x_i\le y_i \text { for } i=1,\ldots,d.
\]
Then all results stated above for the sum distribution $G$ also 
hold for the distribution function $G_\Psi$ of the aggregated 
random variable $\Psi(X)$. 
In particular, the estimator  
$\GPsinast(t):=\ProbFnast\robr{\cubr{x\in\Rd : \Psi(x) \le t}}$ 
converges 
$\Prob$-a.s.\ to $G_\Psi$ in $\linf(\R)$ under the assumptions of 
Theorem~\ref{thm:1} and has convergence rate $O_\Prob(n^{-1/2})$ 
under the assumptions of Theorem~\ref{thm:2}. 
\end{corollary}
\begin{remark} \label{rem:5}
\begin{enumerate}[(a)]
\item
It depends on the aggregation function $\Psi$ whether the generalization 
stated above is advantageous. 
In some special cases even stronger results are possible. 
If, for instance, $\Psi(x)=\max\cubr{x_1,\ldots,x_d}$, then  
the convergence of $\GPsinast$ in $\linf$ is 
related to the uniform convergence of the empirical measure $\ProbCn$ 
on the rectangle set class $\Rcal_d\cap\uintd$. 
As the latter set class is VC, 
one can derive a Donsker Theorem for $\GPsinast$ 
with a precise asymptotic variance. 
\item
{\secondrevision%
Another remarkable example is the \emph{Kendall process}, which is  
obtained by taking the joint distribution function $F$ 
as aggregating function $\Psi$. 
The resulting aggregated distribution function is $H(t):=\Prob(F(X)\le t)$. 
Using the notation from above, this means $H:=G_\Psi$ for $\Psi=F$. 
The aggregated distribution function $H(t)$ 
can be estimated by the empirical distribution 
$H_n:=n\inv\sumjonen 1\cubr{\Fnast(\Xtilde\sampj) \le t}$,
where $\Xtilde\sampj$ are the Iman--Conover synthetic  
variables defined in~\eqref{eq:047} and $\Fnast$ is their empirical distribution function (cf.~\eqref{eq:008}).  
If the margins $\Fi$ are continuous, then $F(X)$ has the same distribution as 
$C(U)$ for $U\sim C$. 
Moreover, $\Fnast(\Xtilde)$ can always be written as $\Cnast(U)$ for $U\sim C$. 
Thus the distribution of the process $\sqrtn(H_n(t)-H(t))$ 
does not depend on the margins $\Fi$. 
In this case asymptotic normality is also available
\citep[cf.][]{van_der_Vaart/Wellner:2007,Ghoudi/Remillard:1998,Barbe/Genest/Ghoudi/Remillard:1996}.%
}%
\item
In the general case, however, extensions of 
Theorems~\ref{thm:1} and \ref{thm:2} indeed go beyond available convergence results
for empirical multivariate distribution functions.
\end{enumerate}
\end{remark}

\subsection{Proofs of auxiliary results} \label{subsec:3.2}
\begin{proof}[Proof of Lemma~\ref{lem:1}]
Denote 
\[
\Wdeltat\upzer := \anbr{ \anbr{T(\At)} + \delta(1,\ldots,1)} \cap \uintd 
\]
and, subsequently,
\[
\Wdeltat\upi := \robrfl{\Wdeltat\upiminusone - 2 \delta e_i} \cap \uintd
,\quad i=1,\ldots,d. 
\]
The notation $A+x$ for $A\subset\Rd$ and $x\in\Rd$ represents a shift 
of the set $A$, i.e., $A+x:=\cubr{a+x : a \in A}$. 
Further, denote 
\[
\Vdeltat\upi := \overline{\Wdeltat\upiminusone \setminus \Wdeltat\upi}
,\quad i=1,\ldots,d. 
\]
The boundaries of $\Vdeltat\upi$ are Lebesgue null sets, because 
any $\anbr{A}$ for $A\subset\uintd$ is Lebesgue-boundary-less.  
Indeed, the construction of $\anbr{A}$ guarantees that if $u\in\anbr{A}$ and 
$v\in\uintd$, then 
$v \le u$ (componentwise)
%
implies $v\in\anbr{A}$. 
Analogously, if $u\in\uintd\setminus\anbr{A}$ and
$v\in\uintd$ with 
$v\ge u$,
%
then $v\in\uintd\setminus A$. 
This monotonicity property allows to cover the boundary $\bnd\anbr{A}$
by $O(\eps^{1-d})$ $d$-dimensional cubes with edge length $\eps$ for any 
$\eps>0$. The total volume of this coverage is $O(\eps)$, so that sending 
$\eps\to 0$  we obtain $\lambda(\bnd\anbr{A}) = 0$. 
This implies that all sets $\Wdeltat\upj$ and $\Vdeltat\upi$ are 
Lebesgue-boundary-less. 
\par
It is obvious that $\Udelta(\Bt)\subset\overline{\Wdeltat\upzer\setminus\Wdeltat\upd}$. 
Moreover, the construction of $\Vdeltat\upi$ entails that 
\[
\overline{\Wdeltat\upzer\setminus\Wdeltat\upd}
=
\bigcup_{i=1}^d \Vdeltat\upi
\]
and $\lambda(\Vdeltat\upi)\le 2\delta$ for all $i$. 
This yields
$\lambda{\Udelta(\Bt)} 
\le 2d\delta$. 
\end{proof}
\begin{proof}[Proof of Lemma~\ref{lem:4}]
\par
According to \citet[Theorem 2.4.1]{van_der_Vaart/Wellner:1996}, 
a set class $\Ccal$ is  $\Prob$-Glivenko--Cantelli if the 
\emph{bracketing number} 
$\Nbr(\eps,\Ccal,L_1(\Prob))$ is finite for any $\eps>0$. 
The number $\Nbr(\eps,\HcalT,L_1(\Prob))$ is the minimal amount of 
so-called \emph{$\eps$-brackets} $[V,W]$ needed to cover $\Ccal$. 
An $\eps$-bracket $[V,W]$ with respect to $L_1(\Prob)$ is a pair of sets 
satisfying $V \subset W$ and $\Prob(W \setminus V) \le \eps$. 
A set class $\Ccal$ is covered by brackets 
$[V_i,W_i]$, $i=1,\ldots,N$, if each $A\in\Ccal$ satisfies 
$V_i \subset A \subset W_i$ for some $i$. 
The criterion cited above is stated in terms of function classes, 
but it easily applies to 
set classes by identifying sets with their indicator functions. 
\par
A sufficient condition for $\Ccal$ to be $\Prob$-Donsker is 
\begin{align} \label{eq:048}
\int_0^\infty \sqrt{\log \Nbr \robrfl{ \eps, \Ccal, L_2(\Prob) }} \dm \eps 
<
\infty
\end{align}
\citep[cf.][Section 2.5.2]{van_der_Vaart/Wellner:1996}.
The distance of two sets $A$ and $B$ in $L_2(\Prob)$ is related to their 
distance in $L_1(\Prob)$ via
\[
d_{L_2(\Prob)} (A,B) 
= 
\norm{ 1_A -1_B }_{L_2(\Prob)} = d_{L_1(\Prob)} ^{1/2}(A,B).
\]
Hence the $L_2(\Prob)$ bracketing entropy condition~\eqref{eq:048} 
is equivalent to 
\begin{align}\label{eq:021}
\int_0^\infty \sqrt{\log \Nbr \robrfl{ \eps^2, \Ccal, L_1(\Prob) }} \dm \eps 
<
\infty.
\end{align}
\par
Part~(\ref{item:lem.4.a}). 
The set class $\HcalT$ is the collection of all $\anbr{T(\At)}$ 
for $t\in\R$ with a fixed $T$. 
Since the sets $\At$ are increasing in $t$, and $T$ is componentwise non-decreasing, 
we have $\anbr{T(\At)} \subset \anbr{T(A_s)}$ for $t\le s$. 
Consequently, $\HcalT$
can be covered by $O(1/\eps)$ brackets of size $\eps$ 
with respect to $L_1(\Prob')$ for any probability measure $\Prob'$ 
on $\uintd$.
The brackets $[V_i,W_i]$ can be chosen as 
$V_i=\anbr{T(A_{t_i})}$, $W_i=\cup_{t<t_{i+1}}\anbr{T(A_t)}$ 
with an appropriate finite sequence $t_1 < \ldots < t_N$. 
If $\Prob'\anbr{T(\At)}$ has jumps, then it may be difficult to choose $t_i$ 
such that $\Prob'(W_i\setminus V_i) = \eps$ for all $i$. 
In this case 
we may have $\Prob'(W_i\setminus V_i) < \epsilon$ for some $i$, 
but the total number of brackets is still $O(1/\eps)$.  
Thus we have 
\begin{align*} 
\Nbr\robr{\eps,\HcalT,L_1(\Prob')}=O(1/\eps),
\end{align*}
and $\HcalT$ is universally Glivenko--Cantelli. 
If the maximal bracket size is $\eps^2$, then one needs $O(1/\eps^2)$ brackets to cover $\HcalT$. This is sufficient for~\eqref{eq:021}, and hence $\HcalT$ is 
universally Donsker.
\begin{remark} \label{rem:8}
It is also easy to show that the set class $\HcalT$ is VC with index $2$. 
As the sets $\anbr{T(\At)}$ are increasing in $t$, they cannot shatter any 
two-point set. Let $u\sampone,u^{(2)}\in\uintd$, and let $t_1,t_2\in\R$ be such 
that for $k=1,2$ $x\sampk\in\anbr{T(\At)}$ is equivalent to $t\ge t_k$. 
Without loss of generality let $t_1\le t_2$. Then $\HcalT$ cannot pick out $\cubr{u^{(2)}}$, and hence $\HcalT$ is VC. 
From here, universal Glivenko--Cantelli and Donsker properties follow 
if we verify \emph{$\Prob'$-measurability} of $\HcalT$ for any 
probability measure $\Prob'$ on $\uintd$ 
\citep[cf.][Definition 2.3.3]{van_der_Vaart/Wellner:1996}. 
This can also be done. 
\par
However, since the VC property of $\Dcaldeltazer$ remains elusive, 
the proofs of Parts~(\ref{item:lem.4.b}) and~(\ref{item:lem.4.c}) 
are based on bracketing entropy.
This is the reason why the proof of Part~(\ref{item:lem.4.a}) presented above  
is also based on bracketing. 
It gives a short preview of the ideas used below.
\end{remark}
\par
Part~(\ref{item:lem.4.c}). 
Fix $\eps> 0$. 
As $[\emptyset, \uintd]$ is a bracket of size $1$ in $L_1(\ProbC)$ covering 
any subset of $\uintd$, we can assume that $\eps <1$. 
For this $\eps$ we define  
\[
\gamma := \frac{\eps}{8 d K(\eps/8d)}
\] 
and 
\[
\delta_1 := (\delta_0 - \gamma)\pospart. 
\]
Given an arbitrary $t_0\in\R$ and $t_1\ge t_0$, consider the following sets:
\begin{align*}
W_{t_0,t_1}
&:=\bigcup_{t\in[t_0,t_1)} U_{\delta_0}(\Bt),\\
W'_{t_0,t_1}
&:=\bigcup_{t\in[t_0,t_1)} \Udeltaone(\Bt),\\
V_{t_0,t_1} 
&:=\bigcap_{t\in[t_0,t_1)} \Udeltaone(\Bt).
\end{align*}
The bracket $[V_{t_0,t_1}, W_{t_0,t_1}]$ covers all $\Udelta(\Bt)$ 
for $\delta\in[\delta_0,\delta_1]$ and $t\in[t_0,t_1)$. 
The size of this bracket in $L_1(\ProbC)$ equals
\[
\ProbC (W_{t_0,t_1} \setminus W'_{t_0,t_1}) 
+ 
\ProbC (W'_{t_0,t_1} \setminus V_{t_0,t_1}).
\]
Denote $J_\eta:=[\eta, 1-\eta]^d$ for $\eta\in[0,1/2)$, and let 
$\lambda$ be the Lebesgue measure on $J_0=\uintd$. Then  
\begin{align} \label{eq:049}
\ProbC(W_{t_0,t_1} \setminus W'_{t_0,t_1}) 
\le 
\ProbC \robrfl{J_0 \setminus J_{\eps/8d} }
+ 
\lambda(W_{t_0,t_1} \setminus W'_{t_0,t_1}) K(\eps/8d). 
\end{align}
As $C$ is a copula and has uniform marginal distributions, 
the first term on the right hand side satisfies
\begin{align*}
\ProbC \robrfl{ J_0 \setminus J_{\eps/8d} }
&\le
\sumioned \ProbC 
\robrfl{ \cubrfl{ u \in \uintd : u_i \notin [\eps/8d, 1-\eps/8d] } } 
=
\eps/4.
\end{align*}
To obtain an upper bound for $\lambda\robr{W_{t_0,t_1} \setminus W'_{t_0,t_1}}$, 
observe that 
\[
W_{t_0,t_1} \setminus W'_{t_0,t_1} 
\subset
\robrfl{
\anbr{U_{\delta_0}(B_{t_1})} \setminus \anbr{U_{\delta_1}(B_{t_1})}}
\cup
\robrfl{
\anbrup{U_{\delta_0}(B_{t_0})} \setminus \anbrup{U_{\delta_1}(B_{t_0})}}
\]
where $\anbrup{B}:=\cup_{u\in B}[u,(1,\ldots,1)]$ is the 
upper layer of $B$ in $\uintd$. 
Moreover, analogously to Lemma~\ref{lem:1}, one can obtain that
\[
\lambda\robr{\anbr{U_{\delta_0}(B_{t_1})} \setminus \anbr{U_{\delta_1}(B_{t_1})}} 
\le 
d(\delta_0-\delta_1) 
\le 
d\gamma 
= 
\frac{\eps}{8K(\eps/8d)}. 
\]
Same bound holds for $\lambda(\anbrup{U_{\delta_0}(B_{t_0})} \setminus \anbrup{U_{\delta_1}(B_{t_0})})$. 
Thus~\eqref{eq:049} yields 
\begin{align} \label{eq:026}
\ProbC(W_{t_0,t_1} \setminus W'_{t_0,t_1}) 
\le 
\frac{\eps}{2}. 
\end{align}
Consequently, as $\ProbC$ is absolutely continuous, 
we can choose $t_1>t_0$ such that either $t_1<\infty$ and 
$\ProbC(W_{t_0,t_1} \setminus V_{t_0,t_1}) = \eps$ 
or 
$\ProbC(W_{t_0,t_1} \setminus V_{t_0,t_1}) < \eps$ and $t_1 = \infty$. 
If $t_1 < \infty$, then~\eqref{eq:026} implies that 
\begin{align}\label{eq:030}
\ProbC(W'_{t_0,t_1} \setminus V_{t_0,t_1}) \ge \frac{\eps}{2}.
\end{align}
Proceeding in the same way as above, we obtain an increasing sequence 
$t_0 < t_1 < t_2 < \ldots$ that eventually terminates at $\infty$. 
Technical difficulties related to possible jumps of 
$\ProbC(W'_{t_i,t} \setminus V_{t_i,t})$ for $t>t_i$ can be handled 
analogously to the proof of Part~(\ref{item:lem.4.a}). 
A similar construction yields a decreasing sequence $t_0>t_{-1}>\ldots$ 
that eventually terminates at $-\infty$.
\par
We still have to show that the sequence $t_k$ is always finite, 
i.e., that $t_k$ indeed assumes $\pm\infty$ for some $k$. 
Consider the sets 
$S_k  := W'_{t_k,t_{k+1}}\setminus \Udeltaone(B_{t_{k+1}})$ and 
$S'_k := W'_{t_k,t_{k+1}}\setminus \Udeltaone(B_{t_k})$.  
As $S_k$ are disjoint for different $k$, 
we have $\sum_{k}\ProbC(S_k) \le 1$ and, analogously,
$\sum_{k}\ProbC(S'_k) \le 1$. 
It is obvious that 
\[
S_k\cup S'_k = W'_{t_k,t_{k+1}} \setminus (\Udeltaone(B_{t_k})\cap\Udeltaone(B_{t_{k+1}})).
\] 
Furthermore, monotonicity of $T$ implies that 
\[
\Udeltaone(B_{t_k})\cap\Udeltaone(B_{t_{k+1}})
=
\bigcap_{t\in[t_k,t_{k+1}]}\Udeltaone(\Bt) 
\subset
V_{t_k,t_{k+1}}.
\]
This immediately yields 
$W'_{t_k, t_{k+1}} \setminus V_{t_k,t_{k+1}} \subset S_k \cup S'_k$. 
Applying~\eqref{eq:030}, we obtain that 
$\ProbC(S_k) + \ProbC(S'_k) \ge \eps/2$ 
if $t_k$ and $t_{k+1}$ are finite. 
%
%
As the sum $\sum_k \robr{\ProbC(S_k) + \ProbC(S'_k)}$ is bounded by $2$, 
the length of the sequence $(t_k)$ is bounded by $4\ceilbr{1/\eps}+4$. 
Possible discontinuities of $\ProbC(W'_{t_k,t}\setminus V_{t_k,t})$ may increase this number at most $\ceilbr{2/\eps}$ additional steps (cf.\ proof of Part~(\ref{item:lem.4.a})). 
\par
Thus we have shown that the set class 
$\cubr{\Udelta(\Bt): t\in\R,\delta\in[\delta_1,\delta_0]}$ 
can be covered by  
$6\ceilbr{1/\eps}+4$ brackets of size $\eps$ in $L_1(\ProbC)$. 
Defining $\delta_k:=(\delta_{k-1}-\gamma)\pospart$ for $k=2,3,\ldots$, 
we reach $0$ after 
$\ceilbr{ \delta_0/\gamma } = \ceilbr{ 8 \delta_0 d K(\eps/8d) / \eps }$ steps.
As the arguments above apply to any interval $[\delta_k,\delta_{k-1}]$, 
we obtain a coverage for the set class $\Dcaldeltazer$ 
and hence an upper bound for the bracketing number:
\[
\Nbr(\eps, \Dcaldeltazer, L_1(\ProbC)) = O(\eps^{-2} K(\eps/8d)).
\]
According to~\eqref{eq:021}, we need to verify that 
\begin{align}
\label{eq:027}
\int_0^1 \sqrt{\log(K(\eps^2/4d) \eps^{-4})} \dm \eps < \infty.
\end{align}
Changing the upper integral bound from $\infty$ to $1$ is 
justified by the fact that any set class can be 
covered by a single bracket of size $1$.
\par
As $\int_0^1\sqrt{\log(1/\eps)} \dm \eps < \infty$ and 
\[
\sqrt{\log(K(\eps^2/4d)) \eps^{-4}} 
\le
\sqrt{\log K(\eps^2/4d)} + 2\sqrt{\log(1/\eps)},
\]
the integrability condition~\eqref{eq:027} follows from the 
assumption~\eqref{eq:025}.
\par
Part~(\ref{item:lem.4.b}). According to proof of part~(\ref{item:lem.4.c}),
$\Nbr\robr{\eps,\Dcaldeltazer,L_1(\ProbC)}<\infty$ for any $\eps$. 
Assumption~\eqref{eq:025} is needed only to verify~\eqref{eq:021}.
\end{proof}

\section{Examples} %
\label{sec:4}
In this section we discuss the regularity 
assumptions of Theorems~\ref{thm:1} and \ref{thm:2}. 
The main results are stated in Propositions~\ref{prop:1} and \ref{prop:2}, 
verifying all regularity assumptions for bivariate Clayton copulas and 
bivariate Gauss copulas with correlation parameter $\rho\ge 0$. 
The case $\rho<0$ is treated in Proposition~\ref{prop:2}(\ref{item:prop.2.c}), 
which guarantees 
the convergence rate $O_\Prob(n^{-1/2}\sqrt{\log n})$. This is almost as good 
as $O_\Prob(n^{-1/2})$. 
\par
We start the discussion with a remark covering the mild 
integrability condition~\eqref{eq:025} and 
copulas with bounded densities. 
\begin{remark} \label{rem:3}
\begin{enumerate}[(a)]
\item \label{item:rem.3.a}
It is easy to see that $K(\eps)=O(\exp(\eps^{-1 + \eta}))$ for $\eps\to0$ 
with some $\eta>0$ implies~\eqref{eq:025}. In particular, any polynomial bound 
$K(\eps) = O(\eps^{-k})$ for $k>0$ is sufficient.
\item \label{item:rem.3.b}
An immediate consequence of Lemma~\ref{lem:1} is that 
all copulas with bounded densities satisfy all regularity 
conditions of Theorems~\ref{thm:1} and \ref{thm:2}. 
A particularly important copula example with a bounded density is the 
\emph{independence copula} $C(u)= \prod_{i=1}^{d}u_i$.  
The Iman--Conover method with independence copula is a standard tool 
in applications with empirically margins based on real data. 
It is applied to generate multivariate samples with margins that are 
close to independent or to remove spurious correlations from multivariate 
data sets. 
\end{enumerate}
\end{remark}
%
%
\par
Unfortunately, many popular copulas, such as Gauss, Clayton, Gumbel, 
or $t$ copulas, have unbounded densities. 
In particular, a bounded copula density implies that all tail dependence 
coefficients are zero. 
Thus applications related to dependence of rare events demand a deeper study 
of copulas with unbounded densities. 
The present paper provides two bivariate examples: the Clayton and the 
Gauss copula. 
\par
The bivariate Clayton copula with parameter $\theta\in(0,\infty)$ is defined as
\[
C_\theta(u_1,u_2) = \robrfl{u_1\powminustheta + u_2\powminustheta -1}^{-1/\theta}
\]
The density $c_\theta$ can be obtained by differentiation:
\begin{align}
c_\theta(u_1,u_2) 
&= \nonumber
\partial_{u_2}\partial_{u_1} C_\theta(u_1,u_2) \\
&= \label{eq:057}
\robrfl{u_1\powminustheta + u_2\powminustheta - 1}^{-2-1/\theta}(\theta+1)(u_1u_2)^{-\theta -1}.
\end{align}
The next result states that this copula family satisfies all 
regularity assumptions of Theorem~\ref{thm:2}. 
\begin{proposition} \label{prop:1}
Any bivariate Clayton copula $C_\theta$ with $\theta\in(0,\infty)$ 
satisfies~\eqref{eq:024} and \eqref{eq:025}.
\end{proposition} 
\begin{proof}
To verify~\eqref{eq:025}, it suffices to show that $K(\eps)$ is polynomial 
(cf.\ Remark~\ref{rem:3}(\ref{item:rem.3.a})). 
The density is given in~\eqref{eq:057}. 
It is easy to see that $\robr{u_1\powminustheta + u_2\powminustheta -1}^{-2 -1/\theta}\le 1$ for $u\in(0,1)^2$. 
Hence the order of magnitude of $K(\eps)$ is determined by $\sup_{u\in(\eps,1-\eps)^2}(u_1u_2)^{-\theta -1}$, which is clearly polynomial. 
\par
To verify~\eqref{eq:024}, recall
that the proof of Lemma~\ref{lem:1} used the following coverage of 
the set $\Udelta(\Bt)$:
\[
\Udelta(\Bt) \subset \bigcup_{i=1}^{d} \Vdeltat\upi.
\]
Hence, for $d=2$, we have    
\begin{align}
\label{eq:039}
\ProbCtheta(\Udelta(\Bt)) \le \ProbCtheta(\Vdeltat\upone) + \ProbCtheta(\Vdeltat\uptwo). 
\end{align}
The arguments that yield $O(\delta)$ bounds for $\ProbCtheta(\Vdeltat\upi)$ are 
symmetric in $i=1,2$, so that it suffices to consider $\Vdeltat\upone$. 
For $u_2\in(0,1)$, denote 
\[
\ubar_1 = \ubar_1(u_2) := \sup\cubrfl{u_1 : (u_1,u_2)\in \Vdeltat\upone}
\]
and 
\[
\ulbar_1 = \ulbar_1(u_2) := \inf\cubrfl{u_1 : (u_1,u_2)\in \Vdeltat\upone}.
\]
It is easy to see that 
\[
\ProbCtheta\robrfl{\Vdeltat\upone} 
= \int_0^1 \int_{\ulbar_1}^{\ubar_1} c_\theta(u_1, u_2) \dm u_1 \dm u_2.
\]
Moreover, the construction of $\Vdeltat\upone$ implies that 
\[
\forall u_2\in(0,1) 
\quad
\ubar_1(u_2) - \ulbar_1(u_2) \le 2\delta 
\quad
\lambda\text{-a.s.}
\]
Partial differentiation of $\log c_\theta$ yields
\[
\partial_{u_i} \log c_\theta(u_1,u_2) 
=
\frac{(2\theta + 1)u_i^{-\theta -1}}{u_1^{-\theta} + u_2^{-\theta} -1} 
-
\frac{\theta + 1}{u_i}
,\quad i=1,2. 
\]
Hence $\partial_{u_1}\log c_\theta(u_1,u_2) = 0$ is equivalent to 
\[
\frac{(2\theta + 1)u_1^{-\theta}}{u_1^{-\theta} + u_2^{-\theta} -1} 
=
\theta + 1,
\]
and for fixed $u_2\in(0,1)$ the copula density $c_\theta(u_1,u_2)$ attains 
its maximum at 
\[
\uast_1 =\uast_1(u_2) := \min\cubrfl{\robrfl{\frac{\theta +1}{\theta}\robrfl{u_2^{-\theta} - 1}}^{-1/\theta} , 1}.
\]
Furthermore, $c_\theta(u_1,u_2)$ is increasing in $u_1$ for $u_1<\uast_1$ and 
decreasing in $u_1$ for $u_1>\uast_1$. Let $\Dplus\upone$ and $\Dminus\upone$ 
denote the corresponding sub-domains of $(0,1)^2$: 
\[
\Dplus\upone:=\cubrfl{u\in(0,1)^2 : u_1< \uast_1}
,\quad
\Dminus\upone:=\cubrfl{u\in(0,1)^2 : u_1> \uast_1}
.
\] 
An exemplary plot of the function $\uast_i$ with resulting sets 
$\Dplus\upone$, $\Dminus\upone$ is given in Figure~\ref{fig:1}. 
Note that the function $u_2\mapsto \uast_1(u_2)$ 
is non-decreasing for any $\theta>0$.   
\par
\begin{figure} 
\centering
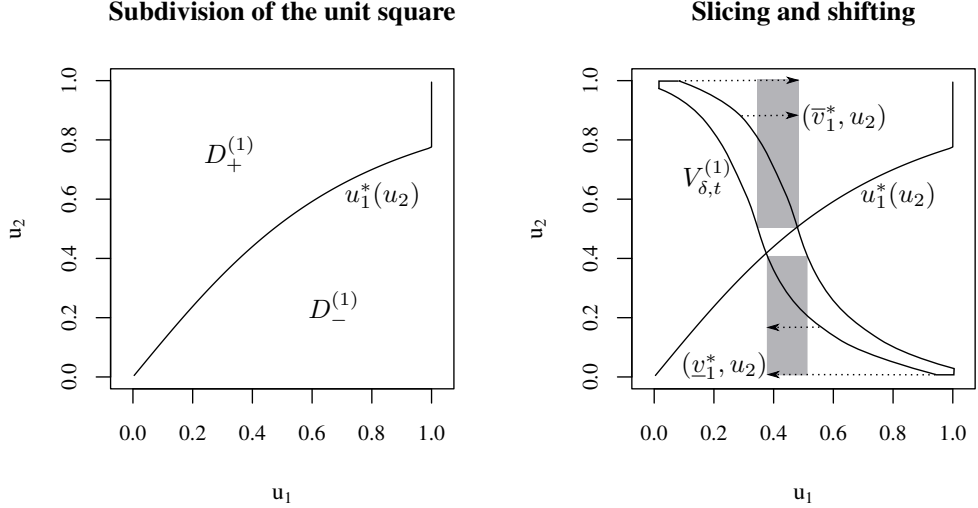
\caption{Bivariate Clayton copula: the curve $u_2\mapsto(\uast_1(u_2),u2)$ subdividing the unit square (left) and the resulting slicing and shifting argument (right).}
\label{fig:1}
\end{figure}
\par
We will show that $\ProbCtheta\robr{\Vdeltat\upone \cap \Dplus\upone}$ and 
$\ProbCtheta\robr{\Vdeltat\upone \cap \Dminus\upone}$ are bounded by $2\delta$. 
Denote
\[
\Iplus\upone := \pi_2\robrfl{\Vdeltat\upone \cap \Dplus\upone}
\quad\text{and}\quad
\Iminus\upone := \pi_2\robrfl{\Vdeltat\upone \cap \Dminus\upone}, 
\]
where $\pi_2$ is the projection on the second coordinate: 
$\pi_2((u_1,u_2)):=u_2$. 
Further denote 
$S_\ast\upone:=\cubr{u\in\Vdeltat\upone : u_1=\uast_1(u_2)}$ and 
$\Iast:=\pi_2(S_\ast\upone)$. 
It is easy to see that $\Iast\upone = \overline{\Iplus\upone} \cap \overline{\Iminus\upone}$. 
Hence we can write
\begin{align}
\label{eq:037}
\ProbCtheta\robrfl{\Vdeltat\upone \cap \Dplus\upone} 
=
\int_{\Iplus\upone \setminus \Iast\upone} 
\int_{\ulbar_1}^{\ubar_1} c_\theta (u)\, \dm u_1 \dm u_2 
+
\int_{\Iast\upone} 
\int_{\ulbar_1}^{\uast_1} c_\theta (u)\, \dm u_1 \dm u_2 
\end{align}
Denote $\vbarast_1 := \sup\cubr{u_1: (u_1,u_2)\in\Vdeltat\upone \cap \Dplus\upone}$. 
This definition implies that 
$\ubar_1(u_2) \le \vbarast_1$ for all  
$u_2\in \Iplus\upone \setminus \Iast\upone$.
Moreover, it is easy to see that 
$\vbarast_1 = \sup\cubr{\uast_1(u_2) : u_2 \in \Iast\upone}$.
As $\uast_1(u_2)$ is non-decreasing, this yields 
$\vbarast_1 \le \uast_1(u_2)$ 
for $u_2\in\Iplus\upone \setminus \Iast\upone$. This gives us
\[
\forall u_2\in \Iplus\upone \setminus \Iast\upone
\quad
\ubar_1(u_2) \le \vbarast_1 \le \uast_1(u_2)   
\]
and, as a consequence,  
$\Vdeltat\upone \cap \Dplus\upone \subset [0,\vbarast_1] \times (0,1)$. 
If $\vbarast_1 \le 2\delta$, then the uniform margins of the copula 
$C_\theta$ immediately yield
$\ProbCtheta\robr{\Vdeltat\upone \cap \Dplus\upone} \le 2\delta$.
Hence, without loss of generality, we assume that $\vbarast_1 > 2\delta$. 
\par 
As $c_\theta$ is increasing in $u_1$ on $\Dplus\upone$ and 
$\ubar_1 - \ulbar_1 \le 2 \delta$ $\lambda$-a.s.,  we obtain that 
\begin{align}
\label{eq:035}
\forall u_2 \in \Iplus\upone \setminus \Iast\upone 
\quad
\int_{\ulbar_1}^{\ubar_1} c_\theta (u)\, \dm u_1 
\le
\int_{\vbarast_1-2\delta}^{\vbarast_1} c_\theta (u)\, \dm u_1
\quad
\lambda\text{-a.s.} 
\end{align}
Moreover, it is easy to see that if 
$u_2\in\Iast\upone$ and 
$u_1\in[\ulbar_1,\uast_1]$, 
then   
$u_1\ge\vbarast_1 - 2\delta$. This yields
\begin{align}
\label{eq:036}
\forall u_2\in \Iast\upone
\quad
\int_{\ulbar_1}^{\uast_1} c_\theta (u) \,\dm u_1
\le
\int_{\vbarast_1-2\delta}^{\vbarast_1} c_\theta (u) \,\dm u_1
.
\end{align}
Combining~\eqref{eq:035}, \eqref{eq:036}, and \eqref{eq:037}, we obtain that
\begin{align}
\label{eq:038}
\ProbCtheta \robrfl{\Vdeltat\upone \cap \Dplus\upone}
\le 
\ProbCtheta \robrfl{\sqbrfl{\vbarast_1 - 2\delta, \vbarast_1} \times [0,1]} = 2\delta.
\end{align}
The latter equality is due to the uniform margins of the copula $C$.
\par
The proof of~\eqref{eq:035} formalizes the idea of slicing the 
set $\Vdeltat\upone$ along $u_1$ for every 
$u_2\in\Iplus\upone \setminus \Iast\upone$ 
and shifting each slice $[\ulbar_1,\ubar_1]\times\cubr{u_2}$ upwards 
along $u_1$ until this slice touches the point $(\vbarast_1,u_2)$  
as illustrated in Figure~\ref{fig:1}.
Since $c_\theta$ is increasing in $u_1$ on $\Dplus\upone$, the transformed 
set has a larger probability under $\ProbCtheta$. 
\par
Using the fact that $c_\theta$ is decreasing in $u_1$ on $\Dminus\upone$, one 
easily obtains the following analogue to~\eqref{eq:035}:
\[
\forall u_2 \in \Iminus\upone \setminus \Iast\upone 
\quad
\int_{\ulbar_1}^{\ubar_1} c_\theta (u)\, \dm u_1 
\le
\int_{\vlbarast_1}^{\vlbarast_1+2\delta} c_\theta (u)\, \dm u_1
\quad
\lambda\text{-a.s.,} 
\]
where 
$\vlbarast_1:=\inf\cubr{u_1: (u_1,u_2)\in\Vdeltat\upone \cap D_{-}\upone}$. 
This result is obtained by slicing $\Vdeltat\upone$ along $u_1$ for 
all $u_2\in\Iminus\upone\setminus\Iast\upone$ and shifting each slice 
$[\ulbar_1,\ubar_1]\times\cubr{u_2}$ downwards along $u_1$ until 
it touches the point $(\vlbarast_1,u_2)$, cf.\ Figure~\ref{fig:1}. 
\par
Similarly to~\eqref{eq:036}, we have that 
\[
\forall u_2\in \Iast\upone
\quad
\int_{\uast_1}^{\ubar_1} c_\theta (u) \,\dm u_1
\le
\int_{\vlbarast_1}^{\vlbarast_1 + 2\delta} c_\theta (u) \,\dm u_1
.
\]
Hence, analogously to~\eqref{eq:038}, we obtain that 
\[
\ProbCtheta\robrfl{\Vdeltat\upone \cap \Dminus\upone} \le 2\delta,
\]
and, consequently, $\ProbCtheta \robr{\Vdeltat\upone} \le 4\delta$. 
\par
As the Clayton copula is symmetric in $u_1$ and $u_2$, we also have 
$\ProbCtheta\robr{\Vdeltat\uptwo} \le 4\delta$. Thus~\eqref{eq:039} 
yields $\ProbCtheta\robr{\Udelta(\Bt)}  \le 8\delta$, 
and condition~\eqref{eq:024} is satisfied.
\end{proof}
\par
The next example is the bivariate Gauss copula $C_\rho$, defined 
as the copula of a bivariate normal 
distribution with correlation parameter $\rho\in(-1,1)$. 
The parameter value $\rho=0$ yields the independence copula, which is
the uniform distribution on the unit square $(0,1)^2$. 
In this case all regularity conditions are satisfied (cf.\ Remark~\ref{rem:3}(\ref{item:rem.3.b})). 
\par
Let $\Phi$ denote the distribution function of the univariate standard 
normal distribution, and for $u_1,u_2\in(0,1)$ let 
$q_i=q_i(u_i):=\Phi\inv(u_i)$ denote the corresponding standard normal 
quantiles. 
Further, let 
$\Sigma=\robrfl{\begin{array}{cc} 1 & \rho \\ \rho & 1 \end{array}}$ 
denote the $2\times2$ correlation matrix corresponding to $\rho$. 
Then the bivariate Gauss copula $C_\rho$ for $\rho\ne 0$ can be written as
\[
C_\rho(u_1,u_2)
= 
\int_{-\infty}^{q_1}\int_{-\infty}^{q_2} 
\det(\Sigma)^{-1/2} \exp\robrfl{-\half x\tr \Sigma\inv x} \dm x_1 \dm x_2,
\]
where $x$ is considered as a bivariate column vector and $x\tr$ is the
transposed of $x$. 
The copula density $c_\rho$ can be obtained by differentiation:
\begin{align*}
c_\rho(u_1,u_2)
&=
\partial_{u_2}\partial_{u_1}C_\rho(u_1,u_2)\\
&=
\det(\Sigma)^{-1/2} \exp\robrfl{ \half q\tr \robrfl{I - \Sigma\inv} q},
\end{align*}
where $q=(q_1,q_2)\tr$ and $I$ is the identity matrix. 
\par
%
%
\par
\begin{proposition} \label{prop:2}
\begin{enumerate}[(a)]
\item \label{item:prop.2.a}
The bivariate Gauss copula $C_\rho$ always satisfies~\eqref{eq:025}. 
\item \label{item:prop.2.b}
If $\rho\ge0$, then $C_\rho$ 
satisfies~\eqref{eq:024}. 
\item \label{item:prop.2.c}
If $\rho<0$, then 
the 
estimate $\Gnast$ of the sum distribution $G$ satisfies
\begin{align} \label{eq:045}
\normfl{\Gnast - G}_\infty  = O_\Prob\robrfl{n^{-1/2}\sqrt{\log n}}. 
\end{align}
\end{enumerate}
\end{proposition}
\begin{proof}
As mentioned above, the case $\rho=0$ is trivial. Hence we assume $\rho\ne0$. 
\par
Part~(\ref{item:prop.2.a})
It is easy to verify that $\half q\tr (I-\Sigma\inv) q \le M q\tr q$ 
for any fixed $\rho\in(-1,0)\cup(0,1)$ with some constant $M=M(\rho)>0$. 
Hence, for $\eps\to 0$, 
\[
K(\eps) = O\robrfl{\exp\robrfl{M \absfl{\Phi\inv(\eps)}^2 }}. 
\]
Let $\phi$ denote the standard normal density:  
$\phi(t) := (2\pi)^{-1/2} \exp(-\half t^2)$. 
It is obvious that 
$
\Phi(t) 
= 
\int_{-\infty}^t \phi(s) \,\dm s < \int_{-\infty}^t -s \phi(s) \,\dm s 
= 
\phi(t)
$
for $t < -1$.
This yields  
\begin{align} \label{eq:043}
\absfl{\Phi\inv(\eps)}
<
\absfl{\phi\inv (\eps)} 
= 
\sqrt{-2 \log (\sqrt{2\pi}\eps)}
\end{align}
for sufficiently small $\eps>0$. 
Hence we obtain that 
\[
K(\eps)
=
O\robrfl{\exp\robrfl{-2 M \log\robrfl{\sqrt{2\pi}\eps}}}
=
O\robrfl{\eps^{-2M}}
,\quad
\eps \to 0. 
\]
Thus condition~\eqref{eq:025} follows from Remark~\ref{rem:3}(\ref{item:rem.3.a}). 
\par
Part~(\ref{item:prop.2.b}). 
The verification of~\eqref{eq:024} for $\rho>0$ is analogous to 
Proposition~\ref{prop:1}. 
Due to 
\begin{align} \label{eq:041}
I-\Sigma\inv 
= 
\frac{1}{1-\rho^2}
\robrfl{\begin{array}{cc} -\rho^2 & \rho \\ \rho &-\rho^2\end{array}},
\end{align}
partial differentiation of $\log c_\rho$ in $u_1$ yields
\[
\partial_{u_1} \log c_\rho(u_1,u_2) 
=
\frac{1}{(1-\rho^2)\phi(q_1)}\robr{-\rho^2 q_1 + \rho q_2}. 
\]
Hence $c_\rho(u)=0$ is equivalent to 
\begin{align} \label{eq:042}
u_1 = \uast_1 (u_2) := \Phi\robrfl{\Phi\inv(u_2) / \rho}.
\end{align}
Moreover, it is easy to see that $c_\rho(u)$ with $\rho>0$ 
is increasing in $u_1$ if $u_1<\uast_1$ and decreasing in 
$u_1$ if $u_1>\uast_1$. 
This is illustrated in Figure~\ref{fig:2}.
As $\uast_1(u_2)$ is increasing, the slicing and shifting argument used in 
the proof of Proposition~\ref{prop:1} applies here, and we obtain that 
$\ProbCrho(\Vdeltat\upone) \le 4\delta$. 
Due to the symmetry of $c_\rho(u)$ in $u_1$ and $u_2$, 
partial differentiation in $u_2$ and the slicing and shifting method 
yield $\ProbCrho(\Vdeltat\uptwo) \le 4\delta$. 
Hence~\eqref{eq:039} gives us $\ProbCrho(\Udelta(\Bt)) \le 8 \delta$. 
\par
\begin{figure} 
\centering
\includegraphics[width=\textwidth]{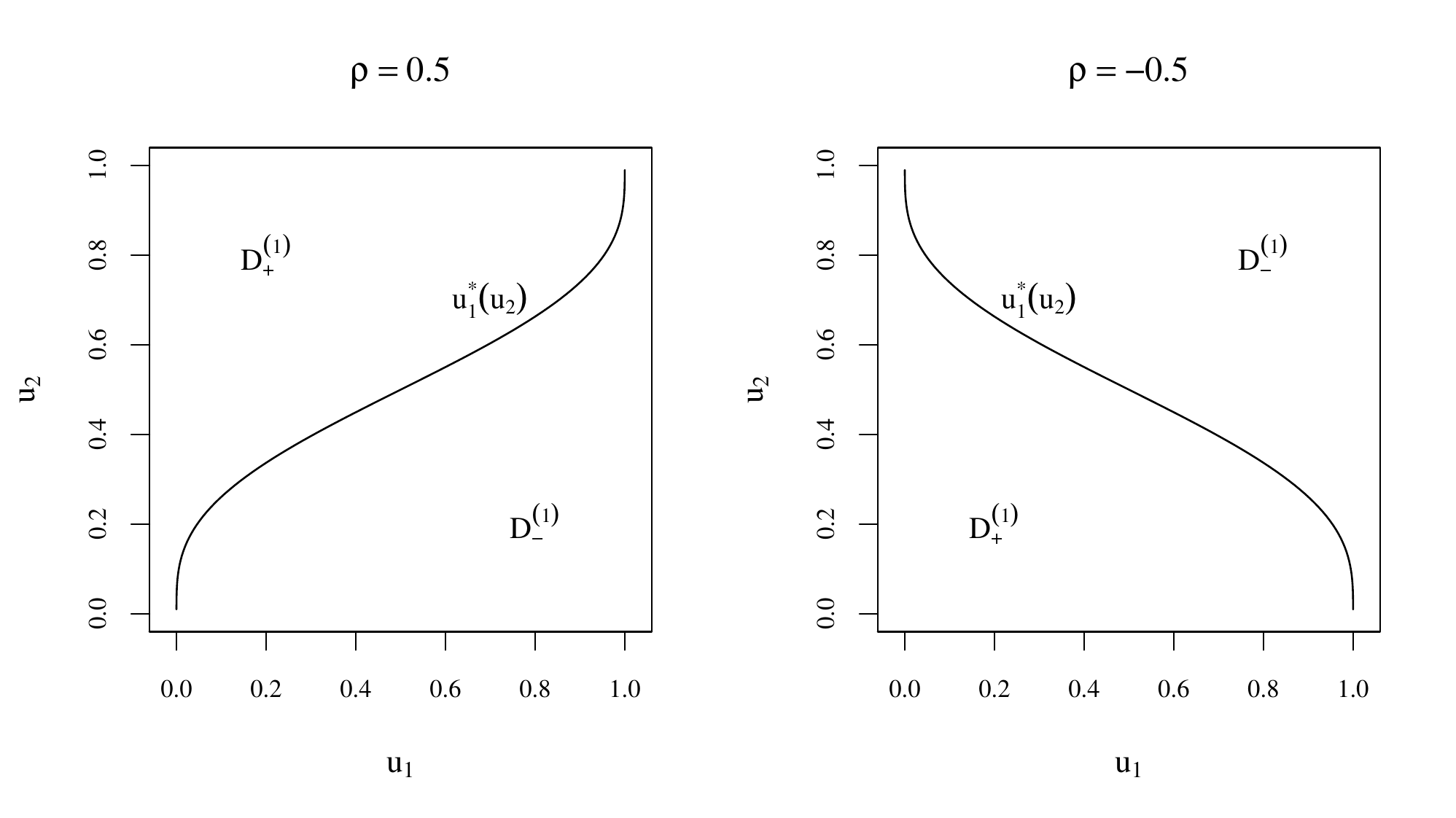}
\caption{The curves $u_2\mapsto(\uast_1(u_2),u_2)$ for the bivariate Gauss copula with correlation $\rho=\pm0.5$.}
\label{fig:2}
\end{figure}
\par
Part~(\ref{item:prop.2.c}). 
The situation for $\rho<0$ is different. 
The copula density $c_\rho(u)$ is still increasing in $u_1$ for $u_1\le\uast_1$ 
and decreasing in $u_1$ for $u_1>\uast_1$, 
but the function $\uast_1(u_2)$ is decreasing (cf.\ Figure~\ref{fig:2}). 
Thus~\eqref{eq:035} does not hold here. 
Instead of shifting each slice $[\ulbar_1,\ubar_1]\times\cubrfl{u_2}$ 
as in the proof of Proposition~\ref{prop:1}, one can shift it until 
it touches the point $(\uast_1,u_2)$. If $\uast_1\in[\ulbar_1,\ubar_1]$, 
no shift is needed. 
That is, we replace the interval $[\ulbar_1,\ubar_1]$ by the interval 
$[\ulbar_1,\ubar_1] + \Delta$, where 
\[
\Delta 
= 
\Delta(u_2)
:= 
\begin{cases}
\uast_1 - \ubar_1 & \text{if } \ubar_1 < \uast_1,\\
\uast_1 - \ulbar_1 & \text{if } \ulbar_1 > \uast_1,\\
0 & \text{else.}
\end{cases}
\]
It is easy to see that 
\[
\forall u_2 \in (0,1)
\quad
\int_{\ulbar_1}^{\ubar_1} c_\rho(u) \, \dm u_1 
\le
\int_{\ulbar_1 +\Delta}^{\ubar_1 + \Delta} c_\rho(u) \, \dm u_1. 
\]
Integrating over $u_2$, we obtain that 
\[
\ProbCrho\robrfl{\Vdeltat\upone} 
\le 
\ProbCrho\robrfl{U^\ast},
\]
where $U^\ast:= \cubr{u\in(0,1)^2: \abs{u_1 - \uast_1(u_2)} < 2\delta}$. 
It is easy to see that  
\[
\cubr{u \in U^\ast : \uast_1(u_2) \notin (2\delta,1-2\delta)}   
\subset 
\robr{[0,4\delta]\cup[1-4\delta,1]}\times[0,1].
\] 
Hence, as $C_\rho$ is a copula and has uniform margins, we obtain that 
\[
\ProbCrho\robrfl{\cubrfl{u \in U^\ast : \uast_1(u_2) \notin(2\delta,1-2\delta)}} \le 8\delta.
\]
Denote the remaining part of $U^\ast$ by $U_0^\ast$: 
\[
U_0^\ast:={\cubr{u\in U^\ast:\uast_1(u_2) \in (2\delta,1-2\delta)}}. 
\]
As $\uast_1(u_2)$ maximizes $c_\rho$ for fixed $u_2$, we have that 
\[
\ProbCrho\robrfl{U_0^\ast} 
\le 
4\delta
\int_{{\uast_1}\inv(1-2\delta)}^{{\uast_1}\inv(2\delta)}
c_\rho(\uast_1(u_2), u_2) \, \dm u_2. 
\]
Applying~\eqref{eq:041} and~\eqref{eq:042} we obtain that 
\[
c_\rho(\uast_1(u_2),u_2) 
= 
\sqrt{1-\rho^2}
\exp\robrfl{\frac{1}{2}\robrfl{\Phi\inv(u_2)}^2}. 
\]
Thus we need an upper bound for the integral
\[
I(\delta)
:=\int_{{\uast_1}\inv(1-2\delta)}^{{\uast_1}\inv(2\delta)} 
\exp\robrfl{\frac{1}{2}\robrfl{\Phi\inv(u_2)}^2}\, \dm u_2. 
\]
The substitution $u_2=\Phi(t)$ yields
\[
I(\delta)
=
\int_{\frac{1}{\rho} \Phi\inv(1-2\delta)}^{\frac{1}{\rho} \Phi\inv(2\delta)} 
\frac{1}{\sqrt{2\pi}} \dm t
=
\frac{\sqrt{2}}{\rho\sqrt{\pi}} \abs{\Phi\inv(2\delta)}.
\]
Finally, applying~\eqref{eq:043}, we obtain that 
\[
I(\delta)  \le \frac{\sqrt{2}}{\rho\sqrt{\pi}} \sqrt{ - 2 \log (2\delta\sqrt{2\pi})}
= 
O\robrfl{\sqrt{\abs{\log \delta}}}
\]
for $\delta\to 0$.  
This implies that   
\begin{align} \label{eq:044}
\ProbCrho\robrfl{\Vdeltat\upone} = 
O \robrfl{\delta\sqrt{\abs{\log \delta}}}.  
\end{align}
Symmetry arguments yield the same order of magnitude for 
$\Vdeltat\uptwo$ and, as a consequence, for $\ProbCrho\robr{\Udelta(\Bt)}$. 
\par
Obviously, \eqref{eq:044} is slightly weaker than $O(\delta)$ in 
assumption~\eqref{eq:024}. 
Revising the proof of Theorem~\ref{thm:2}, we see that~\eqref{eq:024} 
is used to guarantee that the second term 
in~\eqref{eq:022} is $O_\Prob(1)$. 
Hence, replacing~\eqref{eq:024} by~\eqref{eq:044} in Theorem~\ref{thm:2}, 
one obtains~\eqref{eq:045}.
\end{proof}
\par
\begin{remark} \label{rem:4}
\begin{enumerate}[(a)]
\item \label{item:rem.4.a}
It is currently an open question whether the weaker result of 
Proposition~\ref{prop:2}(\ref{item:prop.2.c}) reflects the reality or 
simply arises from the approximations used in the proof.   
However, it should be noted that the case $\rho<0$ is indeed more difficult 
than $\rho>0$. For $\rho<0$ the curve $(\uast_1(u_2),u_2)$ 
may be much closer to the set $B_t$, 
and hence $\ProbCrho(\Udelta(\Bt))$ may be substantially larger than for $\rho>0$. 
In particular, if 
$X_1\sim\Ncal(0,1)$ and $X_2\sim\Ncal(0,1/\rho^2)$, 
then 
\begin{align*}
\cubrfl{(\uast_1(u_2), u_2):u_2\in(0,1)} 
&=
T\robrfl{\cubrfl{x\in\R^2: x_1 + x_2 =0}}\\
&=
B_0\cap(0,1)^2. 
\end{align*}
That is, for $\rho<0$ one can be confronted with the worst 
case when some $B_t$ entirely falls into the area where the copula density 
is at its largest. 
A graphic example to this issue is given on the left hand side of 
Figure~\ref{fig:3}. 
The set $B_0$ in that plot coincides with the set 
$\cubr{\robr{\uast_1(u_2),u_2}: u_2\in(0,1)}$ 
from the right hand side of Figure~\ref{fig:2} ($\rho=-0.5)$. 
Such coincidence is not possible for $\rho>0$ or for the Clayton copula. 
\item
Intuitively speaking, this problem originates from the negative dependence for 
$\rho<0$, where large values of $X_1$ tend to 
be associated with small values of $X_2$ and vice versa. 
As a consequence, the probability 
$\Prob\robr{X_1+X_2 \le 0}$ is influenced by tail events. 
This effect is much weaker for Gauss copulas with $\rho>0$ and, analogously, 
for many other copulas with positive dependence.  
\item
The margins $F_i$ may also influence the convergence of the estimated 
sum distribution function $\Gnast$.  
In insurance and reinsurance applications the components $X_i\sim F_i$ 
are often non-negative. 
In this case $X_1+X_2\le t$ implies that $X_i\le t$ for $i=1,2$, 
so that the tail events have no influence on $G(t)=\Prob\robr{X_1+X_2\le t}$ 
for moderately large $t$. 
The resulting sets $\Bt$ for $t<\infty$ do not contain any internal points of 
the unit square that are close to 
the upper left or to the lower right vertex. 
That is, $u\in\Bt$ and $u_1>1-\eps$ with a small $\eps$ implies $u_2=0$, and
$u_2>1-\eps$ implies $u_1=0$. These $\Bt$ avoid the areas  
where the density of the Gauss copula with $\rho<0$ is at its highest. 
Thus non-negative margins simplify the estimation of the function $G$ 
for the Gauss copula with $\rho<0$ and, analogously, for many 
other copulas with negative dependence. 
An illustration to the different types of sets $B_t$ is given in 
Figure~\ref{fig:3}. 
\end{enumerate}
\end{remark}
%
\section{Conclusions} 
\label{sec:5}
This paper proves that Iman--Conover based estimates for the distribution 
function of the component sum are strongly uniformly consistent, and it provides 
sufficient conditions for the convergence rate $O_\Prob(n^{-1/2})$. 
Besides the component sum, these results hold for any other 
componentwise non-decreas\-ing function. 
The underlying mathematical problem 
goes beyond the classic uniform convergence results for 
empirical copulas. 
In the context of the Iman--Conover method, the primary cause for this technical 
difficulty is the implicit usage of empirical marginal distributions. 
Similar issues also arise in all multivariate models generated by plugging 
empirical margins into an exact copula 
or by plugging exact marginal distributions into 
an empirical copula. 
%
%
The 
marginal transformations involved in these applications complicate
the resulting estimation problem 
in a way that does not allow to establish asymptotic normality of the 
estimated sum distribution. 
Therefore the weaker $O_\Prob(n^{-1/2})$ statement is quite the best 
result one can achieve. 
\par
The results proved here for the Iman--Conover method extend to 
all models 
generated by plugging empirical margins into an exact copula 
or by plugging exact margins into an empirical copula.  
The regularity conditions needed for the convergence rate $O_\Prob(n^{-1/2})$ 
are satisfied for all copulas with bounded densities, 
all bivariate Clayton copulas, and bivariate Gauss copulas with correlation 
parameter $\rho>0$. 
%
The best convergence rate that could be established for the bivariate Gauss 
copulas with $\rho<0$ is $O_\Prob(n^{-1/2}\sqrt{\log n})$.  
This result suggests that 
negative dependence may slow down the convergence of estimates. 
On the other hand, non-negative components $X_i$, 
as typical in insurance applications, may simplify the problem. 
The proof technique used for the bivariate Clayton copula 
applies to bivariate Gauss copulas with $\rho>0$ and 
may also work for other bivariate copulas with positive dependence. 
A straightforward generalization of this method to higher dimensions is, 
however, not feasible.    
Thus the question for the convergence rate in the $d$-variate case 
with $d>2$ and an unbounded copula density 
is currently open. 
\section*{Acknowledgements}
The author would like to thank RiskLab, ETH Zurich, for financial support.
Further thanks for helpful discussions and 
comments are due to Philipp Arbenz, Sara van de Geer,
Jan Beran, Fabrizio Durante, and Paul Embrechts. 
Last but not least, the author thanks the anonymous referees, 
whose valuable comments and suggestions helped to improve this paper. 
%
\mybibstyle
\small{
\bibliography{AuxFiles/mybib-db}
}
\end{document}